\documentclass[journal]{IEEEtran}
\usepackage{amsmath, amsfonts, amssymb, bm, amsthm}
\usepackage[utf8]{inputenc}
\usepackage{enumitem}
\usepackage{color, xcolor, soul, framed}
\usepackage{cancel}
\usepackage[numbers,sort&compress]{natbib}
\usepackage[font=normalfont]{caption} 

\usepackage{subcaption}
\usepackage{booktabs}
\usepackage{mathtools}
\usepackage{graphicx}
\usepackage{hyperref}
\usepackage{varwidth}
\usepackage{amsmath}

\newtheorem{prop}{Proposition}

\usepackage{multirow}
\usepackage{array}
\usepackage{lipsum}

\definecolor{turquoise}{rgb}{.0,.3,1.0}

\newcommand*{\Ibar}[1]{\overline{#1}}

\newcommand{\TT}{\operatorname{T}}

\newcommand{\EV}{\operatorname{E}}
\newcommand{\VAR}{\operatorname{Var}}
\newcommand{\COV}{\operatorname{Cov}}
\newcommand{\DIAG}{\operatorname{diag}}

\hypersetup{
    colorlinks=true,    citecolor=blue,    linkcolor=red,    filecolor=magenta,  urlcolor=blue,    pdftitle={Overleaf Example},
    }

\begin{document}

\title{Parametric and State Estimation of Stationary MEMS-IMUs: A Tutorial}

\author{Daniel~Engelsman, Yair~Stolero, 
        and~Itzik~Klein,~\IEEEmembership{Senior Member,~IEEE}

\thanks{The authors are with the Hatter Department of Marine Technologies, Charney School of Marine Sciences, University of Haifa, Israel.\\ E-mails: \{dengelsm@campus, ystolero@campus, kitzik@univ\}.haifa.ac.il}}


\maketitle

\begin{abstract}
Inertial navigation systems (INS) are widely used in almost any operational environment, including aviation, marine, and land vehicles. Inertial measurements from accelerometers and gyroscopes allow the INS to estimate position, velocity, and orientation of its host vehicle. However, as inherent sensor measurement errors propagate into the state estimates, accuracy degrades over time. To mitigate the resulting drift in state estimates, different approaches of parametric and state estimation are proposed to compensate for undesirable errors, using frequency-domain filtering or external information fusion. Another approach uses multiple inertial sensors, a field with rapid growth potential and applications. The increased sampling of the observed phenomenon results in the improvement of several key factors such as signal accuracy, frequency resolution, noise rejection, and higher redundancy. This study offers an analysis tutorial of basic multiple inertial operation, with a new perspective on the error relationship to time, and number of sensors. To that end, a stationary and levelled sensors array is taken, and its robustness against the instrumental errors is analyzed. Subsequently, the hypothesized analytical model is compared with the experimental results, and the level of agreement between them is thoroughly discussed. Ultimately, our results showcase the vast potential of employing multiple sensors, as we observe improvements spanning from the signal level to the navigation states. This tutorial is suitable for both newcomers and people experienced with multiple inertial sensors. 
\end{abstract}


\section{Introduction}
\IEEEPARstart{N}{owadays}, inertial navigation systems (INS) have become the most commonly used navigation technique, based on an inertial measurement unit (IMU). With the advancement of micro-electro-mechanical systems (MEMS) inertial sensors, their rapid adoption is being seen in many fields due to their small size, low cost, high cost-effectiveness, and low power consumption. By integrating six degree of freedom (6-DoF) measurements from the gyroscopes and accelerometers, IMUs are able to track the navigation states of their mounting platform. 
However, in the absence of continuous external information such as signal delay or reception loss, the navigation solution becomes heavily reliant on the sensors' performance and their characteristic errors.

It is possible to mitigate these errors in two ways: hardware-wise, using high-grade instrumentation with improved technical capabilities; and software-wise, by applying post-processing filtering at the algorithm level. While the former is often costly or simply not feasible, the latter offers a variety of approaches to handle the errors' consequences.
Typically, there are two types of errors associated with these sensors: stochastic noise and deterministic bias. While stochastic errors tend to vary randomly, their variance remains constant when measured again. Conversely, deterministic errors are repeatable over time, but may vary between measurements \cite{woodman2007introduction}. Traditionally, bandpass filters (BPF) address these errors by selectively attenuating specific frequency bands, thus allowing only desirable components to pass through. However, usually in inertial measurements, noise appears to be white, as its power spectral density (PSD) is equally distributed across the bandwidth and tends to overlap the meaningful signal \cite{Groves2013}. 
\\
Another category of its own is the well-known Kalman filter (KF), often viewed as the main workhorse of real-time tracking problems. Using a dedicated fusion mechanism, noisy sensor measurements are combined with a problem-specific dynamic model, enabling predictions of the system states \cite{Titterton2004}. Guaranteeing its optimality requires a good knowledge of the initial states, small non-linearities, and Gaussian noise behavior. 
Inaccurate initialization results in poorer performances and longer convergence times to the true states due to imprecise propagation of the uncertainties, forward in time. 
\\
To that end, parametric estimation methods propose fitting models to observations and then determining their parameters. Given successful modeling, optimality of each estimate is solely dependent on the sample size, thus motivating the need to collect more data. 
Over the past decade, a growing body of literature has explored the advantages of employing multiple MEMS-IMUs for various INS tasks. Notable works include data fusion \cite{bancroft2011data, skog2014open, zhang2020lightweight, patel2021sensor, libero2022unified}, pedestrian navigation  \cite{bancroft2010multiple, skog2014pedestrian, bose2017noise, chen2020deep}, activity recognition \cite{yuan2014localization, rashid2019window, zhang2022online}, integration with GNSS \cite{guerrier2009improving, martin2013new}, sensor arrays \cite{nilsson2016inertial, skog2016inertial, yu2017precise, yu2021novel}, and lately also self-calibration \cite{larey2020multiple, carlsson2021self, engelsman2022learning}.
\\
In this work we develop a simplified system model, with a new perspective on the error relationship to time, and number of sensors. Our main contribution lies in three key aspects: 
\begin{enumerate}[label=(\roman*)]
    \item User-friendly overview: we break down complex estimation topics into step-by-step analysis, accessible to both entry-level learners and advanced practitioners alike. 
    \item Comprehensive validation - we evaluate the effectiveness of our model both analytically and experimentally.
    \item Open-source access - to enable seamless integration, all contents of this work are publicly available on \href{https://github.com/ANSFL/Multiple-MEMS-IMUs-Estimation}{GitHub}.
\end{enumerate}
The rest of the paper is structured as follows: Section \ref{sec:theory} introduces the relevant theory, Section \ref{sec:formulation} formulates the problem, and Section \ref{sec:MTD} presents the methodology used. Section \ref{sec:AnR} presents the results, Section \ref{sec:disc} provides a discussion, and Section \ref{sec:conc} concludes the study.
\section{Inertial Sensor Models} \label{sec:theory}
Inertial sensors are electronic sampling devices that convert continuous physical time signals into a discrete equally spaced measurement denoted by $\tilde{_{( \ )}}$. The gyroscopes, subscript $g$, provide a 3D vector of angular velocities denoted by $\tilde{\boldsymbol{\omega}}_{ib}^b$. The accelerometers, subscript $a$, indicate a 3D vector of linear accelerations, denoted by $\tilde{\textbf{f}}_{ib}^b$. The inherent error sources that contaminate the ground truth (GT) input signal are described by the following linear relationship \cite{groves2015principles}:
\begin{align} \label{eq:sys-error}
\tilde{\boldsymbol{\omega}}_{ib}^b &= ( \textbf{I}_3 + \textbf{M}_g ) {\boldsymbol{\omega}}_{ib}^b + {\textbf{\textit{b}}}_g + \textit{\textbf{w}}_g  \ , \\
\tilde{\textbf{f}}_{ib}^b &= ( \textbf{I}_3 + \textbf{M}_a ){\textbf{f}}_{ib}^b + {\textbf{\textit{b}}}_a + \textit{\textbf{w}}_a \ .
\end{align}
Its linearity consists of a $3 \times 3$ identity matrix, $\textbf{I}_3$, and a gain matrix, $\textbf{M}$, which determine the error slope, a bias term, {\textbf{\textit{b}}}, which indicates the error offset, and a zero-mean white Gaussian noise, $\textit{\textbf{w}}$, describing the sensor noise. Commonly, noise is assumed to be uncorrelated in time and frequency as 
\begin{align}
\textit{w}(t) \sim \mathcal{N}(0, \sigma_w^2 ) \ , \ \textit{w}(t) \perp \textit{w}(t-1) \ \forall \ t \ .
\end{align}
For the majority of typical operational tasks, both gain matrix and biases are not assumed to exhibit time variation, due to
\begin{align} \label{eq:sensorStationary}
\EV[ \, \dot{\textbf{M}}_j \, ] = \mathbf{0}_{3} \ , \  \EV[ \, \dot{\textbf{\textit{b}}}_j \, ] = \mathbf{0}_{3} \ : \ j \in \{ a , \, g \} \, .
\end{align}
Under WSS conditions, most MEMS gyros are insensitive to Earth's rotation rate; thus, their expected output is reduced to 
\begin{align}
\tilde{\boldsymbol{\omega}}_{ib}^b = ( \textbf{I}_3 + \textbf{M}_g ) \, \cancelto{0}{\boldsymbol{\omega}_{ib}^b } \, + \, {\textbf{\textit{b}}}_g + \textit{\textbf{w}}_g  = {\textbf{\textit{b}}}_g + \textit{\textbf{w}}_g \ .
\end{align}
In a similar manner, the stationary accelerometers are expected to output the gravity projection, $\textbf{g}^n$, alongside the error sources
\begin{align} \label{eq:acc_stationary}
\tilde{\textbf{f}}_{ib}^b  & = ( \textbf{I}_3 + \textbf{M}_a ) \textbf{f}_{ib}^b + {\textbf{\textit{b}}}_a + \textit{\textbf{w}}_a  \notag \\
& = -( \textbf{I}_3 + \textbf{M}_a ) \textbf{T}_n^b \textbf{g}^n + {\textbf{\textit{b}}}_a + \textit{\textbf{w}}_a   \ ,
\end{align}
where $\mathbf{T}_{n}^{b}$ denotes a transformation matrix from the navigation frame to the body frame. In the absence of dynamics, it is evident that the error parameters can be readily estimated. For convenience, let us consider a simplified case study of a leveled platform (and sensor), whose x-y body plane is parallel to the north-east plane in the navigation frame:
\begin{align}
\tilde{\textbf{f}}_{ib}^b  = -\begin{bmatrix} m_{xz} \\ m_{yz} \\ 1 \end{bmatrix} \text{g} + {\textbf{\textit{b}}}_g + \textit{\textbf{w}}_a  \approx -\text{g} \cdot \textbf{e}_z + {\textbf{\textit{b}}}_a + \textit{\textbf{w}}_a  .
\end{align}
This way, only the z-axis accelerometer reads non-zero outputs due to gravity, as cross-coupling effects are considered negligible \cite{bekkeng2008calibration}. 
All GT inputs in any of the sensors axes are assumed to be zero, except for the z-axis that is subjected to gravity, given by Kronecker delta, such that residuals are 
\begin{align} \label{eq:res_omega}
\delta \omega_i  = \tilde{\omega}_i - { \omega_i} &= \textit{b}_{g,i} + \textit{w}_{g,i}  \ \, ; \ i \ \in \ \text{x,y,z} \\
\delta \text{f}_i  = \tilde{\text{f}}_i  - { \text{f}_i } &= \textit{b}_{a,i} + \textit{w}_{a,i}  + \text{g} \, \delta_{iz} \ . \label{eq:res_acc} 
\end{align}
It can be seen that both outputs are governed by a constant bias and random noise, provided that wide-sense stationary (WSS) conditions are met (see Appendix \ref{appendix:a}). Estimating the bias using frequency-based techniques guarantees optimality only in theory, given an infinite sampling rate. 
\\
In contrast, the sample mean of the time series provides maximum likelihood estimates (MLE), and thus is our main focus. 
%
Table~\ref{t:NMC} summarizes the relevant notations used in this article. To emphasize the averaging axis, a distinction is made between a single sensor, subscript $s$, and multiple sensors, subscript $m$. Then, we use the \cite{cox1979theoretical} definition to distinguish between the statistic definition from the estimates; while the former refers to the sample characteristics, the latter (see Section~\ref{sec:formulation}), provides predictions of the unknown population parameters.
\begin{table}[t]
\centering
\caption{Nomenclature of notations.}
\renewcommand{\arraystretch}{2.}
\begin{tabular}{c|c|c|c||c|c|}
\multicolumn{1}{c}{} & \multicolumn{3}{c}{Sample (observable)} & \multicolumn{2}{c}{Population (latent)} \\ 
Sensors & Observations & Mean & Statistics & Items & Parameters \\ \hline
$K=1$ & $\tilde{x}_s$ & $\Ibar{x}_s$ & $\hat{\theta}_s$ & \multirow{2}{*}{$x$} & \multirow{2}{*}{$\theta$} \\ \cline{1-4} 
$K>1$ & $\tilde{x}_m$ & $\Ibar{x}_m$ & $\hat{\theta}_m$ & & \\ \hline
\end{tabular} \label{t:NMC}
\end{table}
\subsection{Single sensor model}
Consider a single sensor sampling $\tilde{\textit{z}}_s(t)$, corresponding to both residuals \eqref{eq:res_omega} and \eqref{eq:res_acc}, whose stationary GT values are known to us. To maintain consistency and prevent confusion with the bias terms, the gravity term is omitted in this section only, as it is deterministically known when the sensors are horizontally leveled. Under WSS conditions, let $\textit{z}_s$ denote the GT population and $\tilde{\textit{z}}_s$ represent the sampled portion
\begin{align}
\tilde{\textit{z}}_s(t) = \hat{\textit{b}}_s + \textit{w}_s (t) \ \in \ \mathbb{R}^{N} \ .
\end{align}
\begin{prop} \label{Prop:I}
Let the true signal satisfy $\textit{z} = \textit{b}$, and its noisy sampling be distributed by $\tilde{\textit{z}}_s \sim \mathcal{N}(\hat{\textit{b}}_s, \sigma^2 )$. Then, the sample mean vs. time follows $\hat{\textit{z}}_s \sim \mathcal{N}(\textit{b}, \frac{\sigma^2}{N} )$.
\end{prop}
\begin{proof}
Upon transition to discrete-time $t \mapsto n$, using the definition of white noise, its expected value satisfies 
\begin{align} \label{eq:Mean_white}
\Ibar{\textit{w}}_s := \EV[ \textit{w}_s(t)] \mapsto \frac{1}{N} \sum_{n=1}^N \textit{w}_s [n] = \sum_{n=1}^N \EV \left[ \textit{w}_s [n] \right] = 0 \ 
\end{align}
such that the sample mean, i.e., the mean estimator, $\hat{\textit{z}}_s$, is shown to be unbiased, as given by
\begin{align} \label{eq:Mean_single}
\hat{\textit{z}}_s = \EV[ \tilde{\textit{z}}_s ] = \EV[ \hat{\textit{b}}_s + \Ibar{\textit{w}}_s ] = \EV[ \hat{\textit{b}}_s ] + \cancelto{0}{\EV[ \Ibar{\textit{w}}}_s ] = \textit{b} \ .
\end{align}
Similarly, its variance is given by 
\begin{align} \label{eq:est_var}
\VAR( \, \hat{\textit{z}}_s \, ) = \VAR( \hat{\textit{b}}_s + \Ibar{\textit{w}}_s )_{ \hat{\textit{b}} \perp \Ibar{\textit{w}} } = \cancelto{0}{\VAR( } \hat{\textit{b}}_s ) + \VAR( \Ibar{\textit{w}}_s ) \notag \\ 
= \VAR( \frac{1}{N} \sum_{n=1}^N \textit{w}_s[n] ) = \frac{1}{N^2} \sum_{n=1}^N \VAR( \textit{w}_s[n] ) =  \frac{\sigma^2}{N} \ , 
\end{align}
where $N$ is the number of discrete measurements.
\end{proof}
From Prop.~\eqref{Prop:I}, it follows that given a single stationary sensor, its sample mean exhibits unbiasedness, with a variance error that is reduced inversely proportional to $N$; namely,
\begin{align}
\delta \hat{\textit{z}}_s = \hat{\textit{z}}_s - {\textit{z}}_s \sim \mathcal{N}(0, \frac{\sigma^2}{N} ) \ .
\end{align}

\subsection{Multiple sensors model}
The next case discusses the behavior of the sample mean over $K$ equally aligned sensors, 
which are assumed to be uncorrelated. 
First, let us define the discrete terms for the mean bias and the mean noise, 
\begin{align} \label{eq:Mean_K}
\Ibar{\textit{b}}_m := \frac{1}{K} \sum_{k=1}^{K} \hat{b}_s[k] \hspace{3mm} \text{;} \hspace{3mm} 
\Ibar{\textit{w}}_m := \frac{1}{NK} \sum_{n=1}^N \sum_{k=1}^{K} \textit{w}_s [n, k] \ 
\end{align}
such that the sample mean term over the K-dimensional measurement matrix $\tilde{\textit{z}}_m \in \mathbb{R}^{N \times K} $ is subjected to 
\begin{align}
\hat{\textit{z}}_m = \frac{1}{K} \sum_{k=1}^K \tilde{\textit{z}}_m[k] = \Ibar{\textit{b}}_m + \Ibar{\textit{w} }_m \ \in \ \mathbb{R}^{{N}} .
\end{align}

\begin{prop} \label{Prop:II}
Let the true signal $\textit{z} = \textit{b}$, sampled by $K$ noisy sensors, follow $\tilde{\textit{z}}_s \sim \mathcal{N}(\hat{\textit{b}}_s, \sigma^2 )$. Then, the sample mean over the $K$-dimensional matrix follows $\hat{\textit{z}}_m \sim \mathcal{N}(\textit{b}, \frac{\sigma^2}{NK} )$.
\end{prop} 
\begin{proof}
Further to \eqref{eq:Mean_white}, the sum of $K$ zero-mean Gaussians is zero, and using \eqref{eq:Mean_single}, the mean of $K$ unbiased estimators is also unbiased. Thereby, the expectation of the sample mean is
\begin{align} \label{eq:K_estimator}
\hat{\textit{z}}_m = \EV[ \tilde{\textit{z}}_m ] = \EV[ \Ibar{\textit{b}}_m + \Ibar{\textit{w}}_m ] = \EV[ \Ibar{\textit{b}}_m ] + \cancelto{0}{\EV[ \Ibar{\textit{w}}}_m ] = \textit{b} \ ,
\end{align}
and due to statistical independence, its variance error exhibits
\begin{align}  \label{eq:Var_K}
\VAR( \, \hat{\textit{z}}_m  \, ) &= \VAR( \Ibar{\textit{b}}_m + \Ibar{\textit{w}}_m )_{ \Ibar{\textit{b}} \perp \textit{w}} = \cancelto{0}{\VAR( }  \Ibar{\textit{b}}_m ) + \VAR( \Ibar{\textit{w}}_m ) \notag \\
&= \VAR( \Ibar{\textit{w}}_m ) = \frac{1}{(NK)^2} \sum_{n=1}^N \sum_{k=1}^K \VAR( \textit{w}_s [n,k] ) \notag \\ 
&= \frac{1}{(NK)^2} (NK \sigma^2 ) = \frac{\sigma^2}{(NK)} \ .
\end{align}
\end{proof}
As can be seen from Prop.~\eqref{Prop:II}, the sample mean is also unbiased, but now its variance error decreases by both factors of time ($N$) and number of sensors ($K$), as given by 
\begin{align}
\delta \hat{\textit{z}}_m(t) &= \hat{\textit{z}}_m(t) - {\textit{z}}_m \sim \mathcal{N}(0, \frac{\sigma^2}{NK} ) \ .
\end{align}
To conclude, this subsection presents how bias and noise are embodied in the estimated parameters, which are the sample statistics. As shown below, their exact determination is critical for both system initialization and bias-compensation, with the aim of mitigating potential INS drift.

\section{System model} \label{sec:formulation} 
Traditionally, the navigation states are observed indirectly through minimizing the residual between the true states ${\mathbf{x}}$ and the state estimator $\hat{\mathbf{x}}$, given by 
\begin{equation} \label{eq:residual}
\delta \mathbf{x} = \mathbf{x} - \hat{\mathbf{x}} \ .
\end{equation}
This enables better conditioning of the optimization problem, as inertial measurements better reflect the changes themselves, rather than the total navigation states \cite{roumeliotis1999circumventing}. To maintain consistency with a Kalman filter implementation, all subcomponents are expressed in the navigation frame
\begin{equation}\label{eq_iInsErrorState}
\delta \mathbf{x} = \left[ \begin{array}{ccccc} 
\delta\mathbf{p}^{n} & \delta\mathbf{v}^{n} & \mathbf{\epsilon}^{n} & {\textbf{\textit{b}}}_a & {\textbf{\textit{b}}}_g \end{array} \right]^{\TT} \  \in \mathbb{R}^{15} \ ,
\end{equation}
where $\delta\mathbf{p}^{n}$ is the position error vector, $\delta\mathbf{v}^{n}$ is the velocity error vector, $\boldsymbol{\epsilon}^{n}$ is the misalignment angle errors, $\delta \boldsymbol{f}_{ib}^b \approx {\textbf{\textit{b}}}_a$ is the accelerometer bias residuals, and $\delta \boldsymbol{\omega}_{ib}^b \approx {\textbf{\textit{b}}}_g$ is the gyro bias residuals. 
Under WSS conditions, and assuming that the system is linear time-invariant (LTI), the differential relationship is represented using a linear state-space model \cite{Groves2013}
\begin{align} \label{eq:errorModel}
\delta\dot{\mathbf{x}}(t) &= \mathbf{F} \delta\mathbf{x}(t) + \mathbf{G} \textbf{\textit{w}}(t) \ .
\end{align}
To simplify the formulation and without loss of generality, both body and navigation frames are assumed to coincide, i.e., $\mathbf{T}_{b}^{n} = \mathbf{I}_{3}$, and secondary effects such as the Earth and transport rates are neglected. Thereby, the system matrix which govern the state vector evolution in time, is given by:
\begin{equation} \label{eq_Fmat}
\mathbf{F} = \renewcommand{\arraystretch}{1.15} 
\left[ \begin{array}{ccccc} 
\mathbf{0}_{3} & \mathbf{I}_{3} & \mathbf{0}_{3} & \mathbf{0}_{3} & \mathbf{0}_{3} \\
\mathbf{0}_{3} & \mathbf{0}_{3} & \mathbf{F}_{23}& \mathbf{I}_{3} & \mathbf{0}_{3} \\
\mathbf{0}_{3} & \mathbf{0}_{3} & \mathbf{0}_{3} & \mathbf{0}_{3} & \mathbf{I}_{3} \\
\mathbf{0}_{3} & \mathbf{0}_{3} & \mathbf{0}_{3} & \mathbf{0}_{3} & \mathbf{0}_{3} \\
\mathbf{0}_{3} & \mathbf{0}_{3} & \mathbf{0}_{3} & \mathbf{0}_{3} & \mathbf{0}_{3} 
\end{array}
\right] \ .
\end{equation}
For brevity, all submatrices are $3\times 3$, and $\mathbf{F}_{23}$ denotes the skew-symmetric matrix $[ \, \times \, ]$ of the accelerometer outputs. Due to the negligible product between the axial biases \cite{titterton2004strapdown}, the gravity vector dominates the term, as follows:
\begin{align} \label{eq:skewSym}
\mathbf{F}_{23} = [  -\textbf{f}^n {\times}  ] = [  \textbf{g}^n {\times}  ] = 
\begin{bmatrix}
0 & -\text{g} & 0 \\
\text{g} & 0 & 0 \\
0 & 0 & 0
\end{bmatrix} \ .
\end{align}
The shaping matrix $\mathbf{G}$ relates the stochastic inputs of the process noise vector $\textbf{\textit{w}}$, thus taking into account model imperfections and guaranteeing numerical stability. 
\begin{equation} \label{eq_Gmat} 
\mathbf{G} = \renewcommand{\arraystretch}{1.2}
\left[ \begin{array}{cccc}
\mathbf{0}_{3} & \mathbf{0}_{3} & \mathbf{0}_{3} & \mathbf{0}_{3} \\
\mathbf{I}_{3} & \mathbf{0}_{3} & \mathbf{0}_{3} & \mathbf{0}_{3} \\
\mathbf{0}_{3} & \mathbf{I}_{3} & \mathbf{0}_{3} & \mathbf{0}_{3} \\
\mathbf{0}_{3} & \mathbf{0}_{3} & \mathbf{I}_{3} & \mathbf{0}_{3} \\ 
\mathbf{0}_{3} & \mathbf{0}_{3} & \mathbf{0}_{3} & \mathbf{I}_{3}
\end{array}
\right] \ \text{;} \ \,
\textbf{\textit{w}} = \left[ 
\begin{array}{c} 
\textbf{\textit{w}}_a \\ \textbf{\textit{w}}_g \\ \textbf{\textit{w}}_{ab} \\ \textbf{\textit{w}}_{gb} \\
\end{array} \right] .
\end{equation}
The accelerometers and gyros noise are given by $\textbf{\textit{w}}_a$ and $\textbf{\textit{w}}_g$, and their in-run bias variation is given by $\textbf{\textit{w}}_{ab}$, and $ \textbf{\textit{w}}_{gb}$. 
\\
Being zero-mean white Gaussian, the process noise is assumed to be uncorrelated with the state estimates:
\begin{align} \label{eq:sys_cov}
\COV ( \textbf{\textit{w}}, \delta \mathbf{x}) = \EV[ \textbf{\textit{w}} \delta \mathbf{x}^{\TT} ] = \mathbf{0} \ , \ \EV[ \textbf{\textit{w}} ]=0 \ .
\end{align}

\subsection{Continuous-time propagation} 
Using a state-space model \eqref{eq:errorModel}, the nominal state vector integrated over time interval $\tau=t-t_0$, is given by 
\begin{align} \label{eq:sys_sol}
{\mathbf{x}}(\tau) = \mathbf{\Phi}(\tau) {\mathbf{x}}(t_0) + \int_{t_0}^t \mathbf{\Phi}(\tau) \mathbf{G} \textbf{\textit{w}} (\tau) d\tau \ , 
\end{align}
where the integrated dynamics are used to propagate the initial conditions at time zero $t_0$, using the state-transition matrix
\begin{equation} \label{eq:phi_Fmat} 
\mathbf{\Phi}(\tau) = e^{ \mathbf{F}\tau } = 
\left[\arraycolsep=2.8pt\def\arraystretch{1.5}
\begin{array}{ccccc} 
\mathbf{I}_{3} & \mathbf{I}_{3} \tau & \frac{1}{2} \mathbf{F}_{23} \tau^2 & \frac{1}{2} \mathbf{I}_{3} \tau^2 & \frac{1}{6} \mathbf{F}_{23} \tau^3 \\
\mathbf{0}_{3} & \mathbf{I}_{3} & \mathbf{F}_{23}\tau & \mathbf{I}_{3}\tau & \frac{1}{2} \mathbf{F}_{23} \tau^2 \\
\mathbf{0}_{3} & \mathbf{0}_{3} & \mathbf{I}_{3} & \mathbf{0}_{3} & \mathbf{I}_{3} \tau \\
\mathbf{0}_{3} & \mathbf{0}_{3} & \mathbf{0}_{3} & \mathbf{I}_{3} & \mathbf{0}_{3} \\
\mathbf{0}_{3} & \mathbf{0}_{3} & \mathbf{0}_{3} & \mathbf{0}_{3} & \mathbf{I}_{3} 
\end{array} 
\right] .
\end{equation}
The autocorrelation function (ACF) describes the correlation of a signal with its delayed copy over time-lag $\tau$. Being uncorrelated, the non-zero value is obtained only for zero-lag,
\begin{align} \label{eq:AutoCorr}
&\mathbf{R}_{xx}(\tau) = \EV[ \textbf{\textit{w}}(t) \textbf{\textit{w}}(t+\tau)^{\TT}] = \COV \left( \textbf{\textit{w}}(t), \textbf{\textit{w}}(t+\tau) \right)
\delta(\tau) \Big|_{\tau=0} \notag \\[2mm] 
& \hspace{2mm} = \VAR( \textbf{\textit{w}} )  = \DIAG \left(  [ \, \mathbf{1}_3^{\TT} \sigma_a^2 , \, \mathbf{1}_3^{\TT} \sigma_g^2 , \, \mathbf{1}_3^{\TT} \sigma_{ab}^2 , \, \mathbf{1}_3^{\TT} \sigma_{gb}^2 ] \, \right).  
\end{align} 
For simplicity, noise is assumed to be isotropic, i.e., equally intense in all three sensor axes \cite{jazwinski2007stochastic}. Next, using the Wiener-Khinchin theorem, the PSD matrix of the process noise is obtained by the Fourier transform, 
\begin{align} \label{eq:PSD}
\mathbf{S}_{xx} &= \int_{t_0}^t \mathbf{R}_{xx}(\tau) e^{-j \omega \tau} d \tau \ , %
\end{align}
exhibiting frequency-invariance due to its whiteness, thus enabling the continuous process noise covariance by (detailed derivation in Appendix \ref{appendix:b})
\begin{align} \label{eq:processCov}
\mathbf{Q}(\tau) &= \EV[ \mathbf{G} \textbf{\textit{w}}(\tau) \textbf{\textit{w}}(\tau)^{\TT} \mathbf{G} ] \notag \\ 
&= \int_{t_0}^t \mathbf{\Phi}(\tau) \mathbf{G} \, \mathbf{S}_{xx} \, \mathbf{G}^{\TT} \mathbf{\Phi}(\tau)^{\TT} d\tau \ . 
\end{align}

\subsection{State observer}
Upon transition from parametric estimation to state estimation, the GT population now refers to the entire latent system model. Also here, $s$ and $m$ subscripts are used to emphasize the number of sensors. 
Let us establish the time relation between our state estimates and the derivative of the true system parameters, 
\begin{align}
\EV[ \dot{\mathbf{x}}(t) ] = \frac{\partial}{\partial t} \hat{\mathbf{x}}(t) \ \Rightarrow \ \dot{\hat{\mathbf{x}}}(t) = \mathbf{F} \, \hat{\mathbf{x}}(t) \ ,
\end{align}
introducing a homogeneous solution by
\begin{align} \label{eq:observer}
\hat{\mathbf{x}}( \tau ) = \EV[ \mathbf{x}( \tau ) ] = \mathbf{\Phi}(\tau) \hat{\mathbf{x}}(t_0) \ .
\end{align}
By subtracting \eqref{eq:observer} from \eqref{eq:sys_sol}, we define the following residual error vector
\begin{align} \label{eq:deltaSol}
\delta {\mathbf{x}} ( \tau ) = {\mathbf{x}}( \tau ) - \EV[ \mathbf{x}( \tau ) ] = \mathbf{\Phi}(\tau) \delta {\mathbf{x}} (t_0) + \mathbf{G} \textbf{\textit{w}}(\tau) \ ,
\end{align}
whose corresponding covariance matrix is given by
\begin{align}
\mathbf{P}(\tau) &= \EV \bigg[ \, \delta {\mathbf{x}}(\tau) \, \delta  {\mathbf{x}}(\tau)^{\TT}  \bigg]_{{ \delta x \perp \textit{w}}}   \\
&= \EV \bigg[ \mathbf{\Phi}(\tau) \delta {\mathbf{x}} (t_0) \delta {\mathbf{x}} (t_0)^{\TT} \mathbf{\Phi}(\tau)^{\TT} + \mathbf{G} \textbf{\textit{w}}(\tau) \textbf{\textit{w}}(\tau)^{\TT} \mathbf{G}^{\TT} \bigg] \ . \notag
\end{align}
After rearranging and simplifying, the continuous error state covariance is given by
\begin{align} \label{eq:errorCov} 
\mathbf{P}(\tau) = \mathbf{\Phi}(\tau)  \mathbf{P}_0 \, \mathbf{\Phi}(\tau)^{\TT} + \mathbf{Q}(\tau) \ .
\end{align}
More on its derivation and relation to the well-known Riccati and Lyapunov equations is in \cite{mitter2005information}.
\begin{prop} \label{Prop:III}
Let $\delta \mathbf{x}(t)$ be a random variable for $\forall \, t>0$. If its a priori estimate is given by $\delta{\mathbf{x}}_0 \sim \mathcal{N} \left( \boldsymbol{\mu}_0, \mathbf{P}_0 \right) $, then the time sequence $\{ \delta \mathbf{x}(t) \}$ follows
$$
\delta{\mathbf{x}}(t) \sim \mathcal{N} \left( \mathbf{\Phi}(\tau)\boldsymbol{\mu}_0, \mathbf{P}(t) \right) .
$$ 
\end{prop}
\begin{proof}
Since time propagation of $\{ \delta x(t_0), ..., \delta x(t_n) \}$ is given by matrix $\mathbf{\Phi}(\tau)$, and based on Gaussianity of all variables, the joint distribution of sequence $\{ \delta \mathbf{x}(t) \}$ must obey a multivariate normal distributions, for a set of time points $\{t_0, ..., t_n \}$.
\end{proof}
By utilizing Prop.~\eqref{Prop:III}, we establish the assumption that the evolution of all system states, including initial conditions and system dynamics, follow a linear Gaussian process. 
\subsection{Initial conditions}
To study how initial conditions, subscript $0$, are triggered in time, the Gaussian nature of \eqref{eq:deltaSol} is closely examined. 
%
\subsubsection{Bias effect} Being static, all initial kinematics can be constrained to zero \cite{klein2010pseudo, engelsman2023information}, leaving solution to propagate deterministically, depending only on the estimated biases.
\begin{equation} \label{eq:delta_gauss}
\EV[ \delta \hat{\mathbf{x}}( \tau ) ] = \mathbf{\Phi} (\tau ) \delta \hat{\mathbf{x}}( t_0 ) = 
\bigg[ \mathbf{\Phi} (\tau ) \bigg] \left[ \begin{array}{c} 
\cancelto{0}{\delta \hat{\mathbf{p}}}_0^{n} \\ \cancelto{0}{\delta \hat{\mathbf{v}}}_0^{n} \\ \cancelto{0}{\, \hat{\mathbf{\epsilon}}_0}^n \\ \hat{{\textbf{\textit{b}}}}_a \\ \hat{{\textbf{\textit{b}}}}_g
\end{array} \right] \ .
\end{equation}
But since the bias states are assumed constant \eqref{eq:sensorStationary}, the only relevant states to track in \eqref{eq:delta_gauss} are the kinematics, such that
\begin{equation} \label{eq:errExpected}
\EV \def\arraystretch{1.4} \left[
\begin{array}{c} 
\delta \hat{\mathbf{p}}^{n}(\tau) \\ \delta \hat{\mathbf{v}}^{n}(\tau) \\ \hat{\mathbf{\epsilon}}^{n}(\tau)
\end{array} \right] = 
\arraycolsep=1.5pt\def\arraystretch{1.7}
\left[ \begin{array}{ccc} 
\frac{1}{2} \mathbf{I}_{3} \tau^2 & \frac{1}{6} \mathbf{F}_{23} \tau^3 \\
\mathbf{I}_{3}\tau & \frac{1}{2} \mathbf{F}_{23} \tau^2 \\
\mathbf{0}_{3} & \mathbf{I}_{3} \tau \\ \end{array} \right] 
\left[ \begin{array}{c} 
\hat{\textbf{\textit{b}}}_a \\ \hat{\textbf{\textit{b}}}_g \end{array} \right] \ .
\end{equation}
As shown, some state predictions propagate quadratically and cubically in time, with respect to the initial sensors residuals; thereby, the single sensor state vector can be abbreviated to
\begin{equation} \label{eq:errExpectedShort}
\EV [ \delta \hat{\mathbf{x}}_{s} ( \tau ) ] = \mathbf{\Phi}_{\text{kin.}} (\tau ) \ \hat{{\textbf{\textit{b}}}}_{0, s} \ ,
\end{equation}
where the subscript kin. denotes the kinematic states. Similarly, given multiple sensors, the mean error vector is shortened to
\begin{equation} \label{eq:errExpectedShort_m}
\EV [ \delta \hat{\mathbf{x}}_{m} ( \tau ) ] = \mathbf{\Phi}_{\text{kin.}} (\tau ) \ \hat{{\textbf{\textit{b}}}}_{0, m} \ .
\end{equation}
\begin{prop} \label{Prop:IV}
When averaged over $K$ identical and equally aligned sensors, the expectation of the mean error vector \eqref{eq:errExpectedShort_m} decreases inversely proportional to their sum.
\end{prop}
\begin{proof}
Further to \eqref{eq:Mean_K}, using $K$ sensors, the mean estimate of the initial bias is reduced inversely proportional to $K$:
\begin{align}
\hat{{\textbf{\textit{b}}}}_{0, m} = \Ibar{{\textbf{\textit{b}}}}_{0, m} = \frac{1}{K} \sum_{k=1}^{K} \hat{{\textbf{\textit{b}}}}_{0, s} [k] \ .
\end{align}
Being linear, system \eqref{eq:errExpectedShort_m} is homogeneous under scaling, such that both $K$ and $\mathbf{\Phi}_{\text{kin.}}$ can be factored out as follows:
\begin{align}
\EV [ \delta \hat{\mathbf{x}}_{m} ( \tau ) ] = \EV [ \sum_{k=1}^K \delta \hat{\mathbf{x}}_{s} ( \tau )  ] = \frac{1}{K} \mathbf{\Phi}_{\text{kin.}} ( \tau ) \sum_{k=1}^K \hat{{\textbf{\textit{b}}}}_{0, s} [k] \ .
\end{align}
In the same manner, it can be further generalized to \eqref{eq:errExpected}, while maintaining all other conditions the same.
\end{proof}
According to Prop.~\eqref{Prop:IV}, in the absence of dynamic inputs, error states are shown to be reduced proportionally to 1/$K$, in correspondence to the reduced bias estimates.
\subsubsection{Noise effect} 
The random variations over time, originating in the stochastic inputs, tend to cancel out each other due to the law of large numbers. Yet, their characterization provides valuable insight about the uncertainty of an individual realization, in terms error of dispersion. The initial values of \eqref{eq:errorCov} lie in $\mathbf{P}_0$, but are invariant to the sensors' errors. Their estimation is mostly performed through careful heuristics, relying on knowledge and intuition of the system designers \cite{verhaegen1986numerical}. 
In contrast, the process noise covariance propagates the PSD matrix in matrix, which embodies the spectral intensity and correlation of the process noise vector. 

\begin{prop} \label{Prop:V}
When averaged over $K$ identical and equally aligned sensors, the covariance of the process noise \eqref{eq:processCov}, decreases inversely proportional to $K$.
\end{prop}
\begin{proof}
Let an estimate of a time series follow $\hat{\boldsymbol{\epsilon}}_{s} \sim \mathcal{N}(0, \sigma^2 )$. Further to \eqref{eq:Var_K}, given $K$ sensors, 
the variance error of their sample mean is reduced by
\begin{align}
\VAR( \hat{ \boldsymbol{\epsilon}}_m ) = \VAR \left( \frac{1}{K} \sum_{k=1}^K  \hat{\boldsymbol{\epsilon}}_s [k] \right) = \frac{1}{K^2} (K \sigma^2 ) = \frac{\sigma^2}{K} \ , 
\end{align}
leading to the following equality
\begin{align} \label{eq:VarIdentity}
\VAR( \hat{ \boldsymbol{\epsilon} }_m ) = \frac{1}{K} \VAR ( \hat{ \boldsymbol{\epsilon} }_s ) \ , 
\end{align}
such that its overall distribution follows in fact
\begin{align}
\hat{ \boldsymbol{\epsilon} }_m \ \sim \ \mathcal{N}(0, \sigma^2/K ) . 
\end{align}
Next, to generalize it to the entire state-space, consider a process noise vector $\textbf{\textit{w}}_s \in \mathbb{R}^{15}$, drawn from an uncorrelated multivariate normal distribution. For a single sensor realization, its error covariance matrix is
\begin{align}
\COV \left( \hat{\textbf{\textit{w}}}_s, \hat{\textbf{\textit{w}}}_s \right) = \DIAG \left( \VAR ( \hat{\textbf{\textit{w}}}_s ) \right) \ ,
\end{align}
however, by averaging each state variable over $K$ sensors, the entire noise covariance follows identity \eqref{eq:VarIdentity}, leading to
\begin{align}
\COV \left( \hat{ \textbf{\textit{w}}}_m, \hat{ \textbf{\textit{w}}}_m \right) = \DIAG \left( \VAR ( \hat{ \textbf{\textit{w}}}_m ) \right) = \frac{1}{K} \DIAG \left( \VAR ( \hat{\textbf{\textit{w}}}_m ) \right) \ .
\end{align}
Substituting in the autocorrelation matrix \eqref{eq:AutoCorr} and further integration to \eqref{eq:processCov}, allows factoring out $K$, as all operators are linear; thereby eliciting a reduced process noise covariance
\begin{align}
\hat{ \mathbf{Q} }_m(\tau) = \ ... \ = \frac{1}{K} \hat{\mathbf{Q}}_s (\tau) \ .
\end{align} 
\end{proof}
Similarly to before, Prop.~\eqref{Prop:V} demonstrates how state uncertainties and covariances can be reduced once estimates are averaged over an array of $K$ sensors.
\\
To conclude this section, the relevant analytical framework is presented in the form of a linear differential system, propagated in time by two triggering factors: deterministic (bias), and stochastic (noise). 
These in turn, were shown to induce a polynomial error growth that drifts the position error \eqref{eq:errExpected} in the order of $\mathcal{O}(\tau^3)$, and similarly the state covariance \eqref{eq:CovarFull}, up to $\mathcal{O}(\tau^7)$. \\

\section{Performance metrics and experimental setup} \label{sec:MTD} 
In this section we outline the procedures and measures used to gather empirical observations for this work, and assess the validity of our analytical analysis.

\subsection{Performance metrics}
To assess our hypothesis, the following functions are used to project the outcomes into a relevant measurable space, thus providing a clearer interpretation:

\begin{itemize}
\item Root mean square (RMS): quantifies the effective magnitude of a signal through the square root of the mean square components. For example, let ${x}\sim \mathcal{N}(\mu, \sigma^2)_{ \mu \perp w }$
\begin{equation} \label{eq:rms}
\text{RMS} = \sqrt{ \EV[x^2]} = \sqrt{\frac{1}{N}\sum_n^N {x}[n]^2 } = \sqrt{\mu^2 + \sigma^2} \ .
\end{equation}
\item Mean squared error (MSE): evaluates the performance of an estimator by calculating the average squared residuals. Being unbiased, i.e., $\EV[ \hat{\mu} ] = \mu$, the MSE of the estimators simplifies to variance error only,
\begin{equation} \label{eq:mse}
\text{MSE} = \EV[ (\hat{\mu}-\mu)^2 ]_{ \hat{\mu} \perp \mu } = \VAR( \hat{\mu} ) \ .
\end{equation}
\item Fisher information: measures the amount of information an observable random variable carries about unknown parameters of the underlying model. In the case of Gaussian distribution, as here, $\ell(x;\mu)$ denotes the log-likelihood function of sample $x$, resulting in \cite{bendat2011random}
\begin{align} \label{eq:FIM}
\mathcal{I}(\mu) = 
- \EV \left[ \frac{\partial^2}{\partial \mu^2} \ell (\boldsymbol{x};\mu) \right] = ... = \frac{N}{\sigma^2} \ .
\end{align}
In the scalar unbiased case, the Cramér–Rao lower bound (CRLB) is obtained by the following reciprocal, 
\begin{align} \label{eq:CRLB}
\VAR( \, \hat{\mu} \, ) \geq \frac{1}{ \mathcal{I}(\mu) } \ \Rightarrow \ \VAR( \, \hat{\mu} \, ) \geq \frac{\sigma^2}{N} \ ,
\end{align}
thus providing a lower bound on the variance error of $\hat{\mu}$, depending on the number of $N$ measurements. 
\end{itemize}

\begin{figure}[h]
\centering 
\includegraphics[width=.35\textwidth, clip, keepaspectratio]{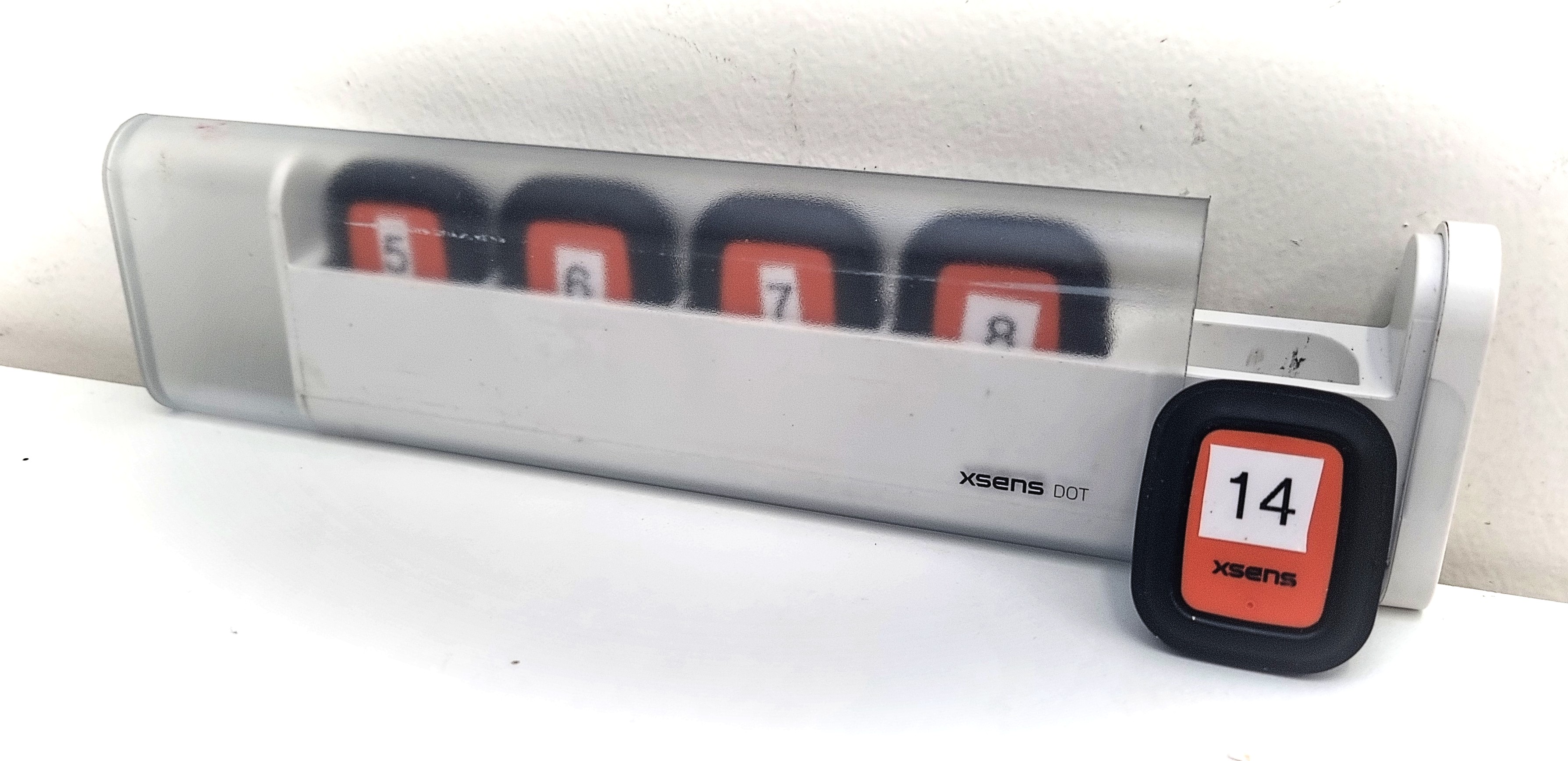}
\caption{The \href{https://www.xsens.com/hubfs/Downloads/DOT/Documents/2021-07\%20-\%20Archived\%20-\%20Xsens\%20DOT\%20User\%20Manual\%20.pdf}{Xsens-DOT}, a dedicated apparatus for alignment and synchronization of five inertial sensors \cite{XsensDot}.}
\label{fig:Sensors}
\end{figure}

\subsection{Experimental setup}
Under laboratory conditions, data from the stationary sensors was acquired in a time span of about 100 seconds, sampled at 100 Hz. Overall, recordings were obtained from a total of ten sensors, using two Xsens-Dot cases \footnote{For reproducibility, both data and code are made publicly available @ \underline{\url{https://github.com/ANSFL/Multiple-MEMS-IMUs-Estimation}}.} carrying five identical sensors each, as shown in Fig.~\ref{fig:Sensors}. 
\subsection{Instrumental errors}
The three main error factors to affect the inertial measurements stability are turn-on bias, in-run bias, and random noise. While the turn-on bias refers to the constant bias obtained when the sensor is powered on, the in-run bias describes its run-time variation. In the mid-range time frame, the Xsens in-run bias accumulated a tenth of the static bias after 100 seconds, thus is considered negligible.\\ 
Table~\ref{t:SensorErrors} summarizes the error variability between all ten sensors, using RMS between all 3D components to guarantee quantities that are orientation-independent. 
\begin{table}[h]
\centering
\caption{Error range of ten sensors @ 100 [Hz].}
\renewcommand{\arraystretch}{1.8}
\begin{tabular}{c c|c|c|c|c|}
& Error & Min & Median & Max & Units \\ \cline{2-6}
\multirow{2}{*}{\rotatebox[origin=c]{90}{\textit{Gyro.}}} & $\hat{{\textit{b}}}_{g,\text{RMS}}$ & 1.987 & 2.164 & 2.343 & \multirow{2}{*}{[deg/s]} \\ \cline{3-5}
 & $\hat{{\sigma}}_{g,\text{RMS}}$ & 0.026 & 0.033 & 0.038 & \\ \cline{2-6}
\multirow{2}{*}{\rotatebox[origin=c]{90}{\textit{Accel.}}} & $\hat{\textit{b}}_{a,\text{RMS}}$ & 0.176 & 0.181 & 0.197 & \multirow{2}{*}{[m/s$^2$]} \\ \cline{3-5}
 & $\hat{\sigma}_{a,\text{RMS}}$ & 0.007 & 0.007 & 0.009 & \\ \cline{2-6}
\end{tabular} \label{t:SensorErrors}
\end{table} 
\\
\section{Analysis and Experimental Results} \label{sec:AnR}
This section presents the experimental outcomes, and analyzes their agreement with the hypothesized model. First, the micro-level of the signal is examined, followed by a macro-level investigation of the entire INS solution.
\begin{figure}[b]
\centering 
\includegraphics[width=.5\textwidth, clip, keepaspectratio]{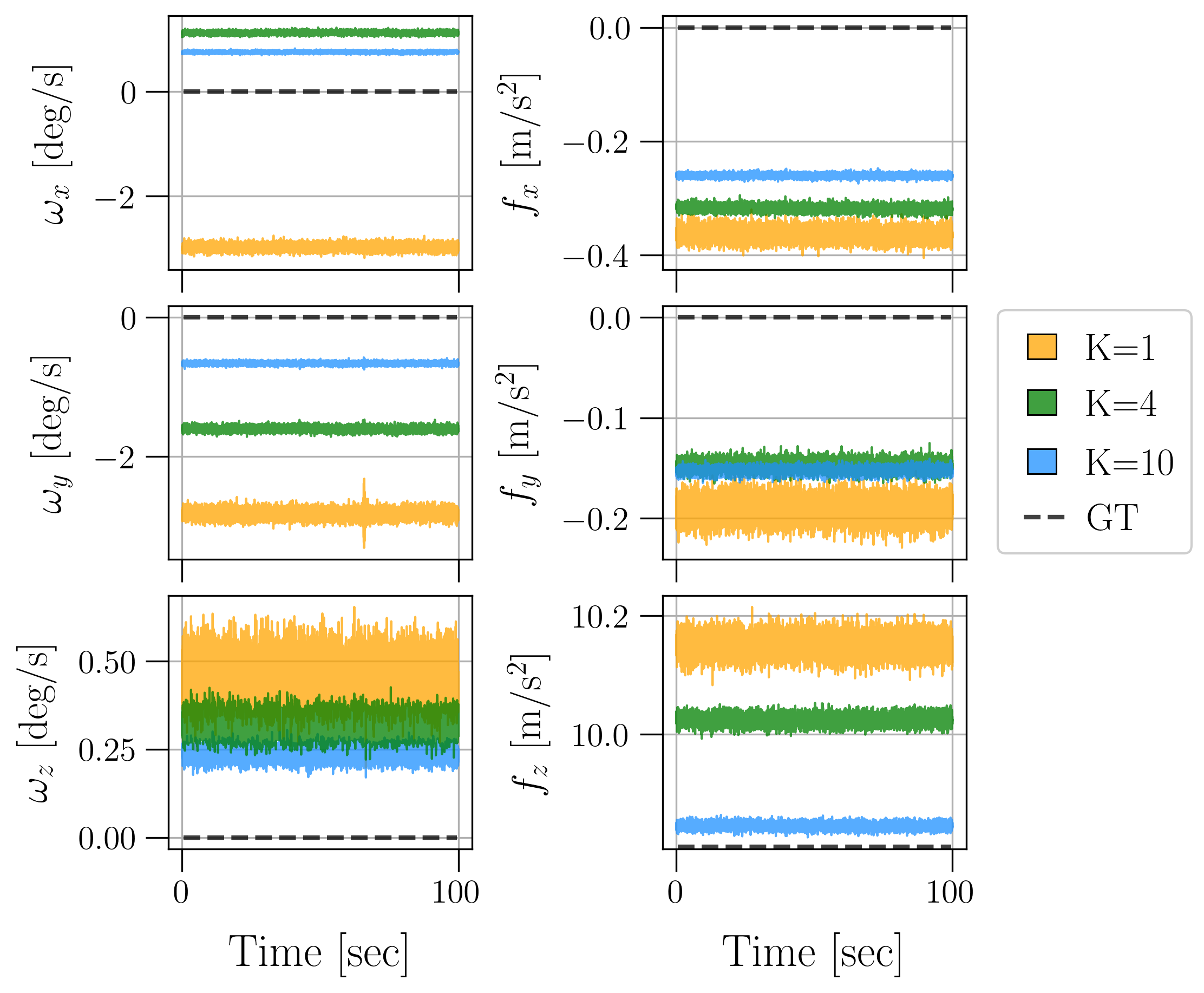}
\caption{Biased measurements of the inertial sensors;\\Left: gyroscopes. Right: accelerometers.}
\label{fig:Meas_Uncalibrated}
\end{figure}

\subsection{Parametric estimation}
The estimated characteristic errors are averaged as a function of the number of sensors, denoted by $K$. To establish a rational arrangement, sensors are sorted in descending order according to their quality, such that K=1 considers measurements from the least reliable sensor.
\\
\begin{figure}[t]
\centering 
\includegraphics[width=.5\textwidth, clip, keepaspectratio]{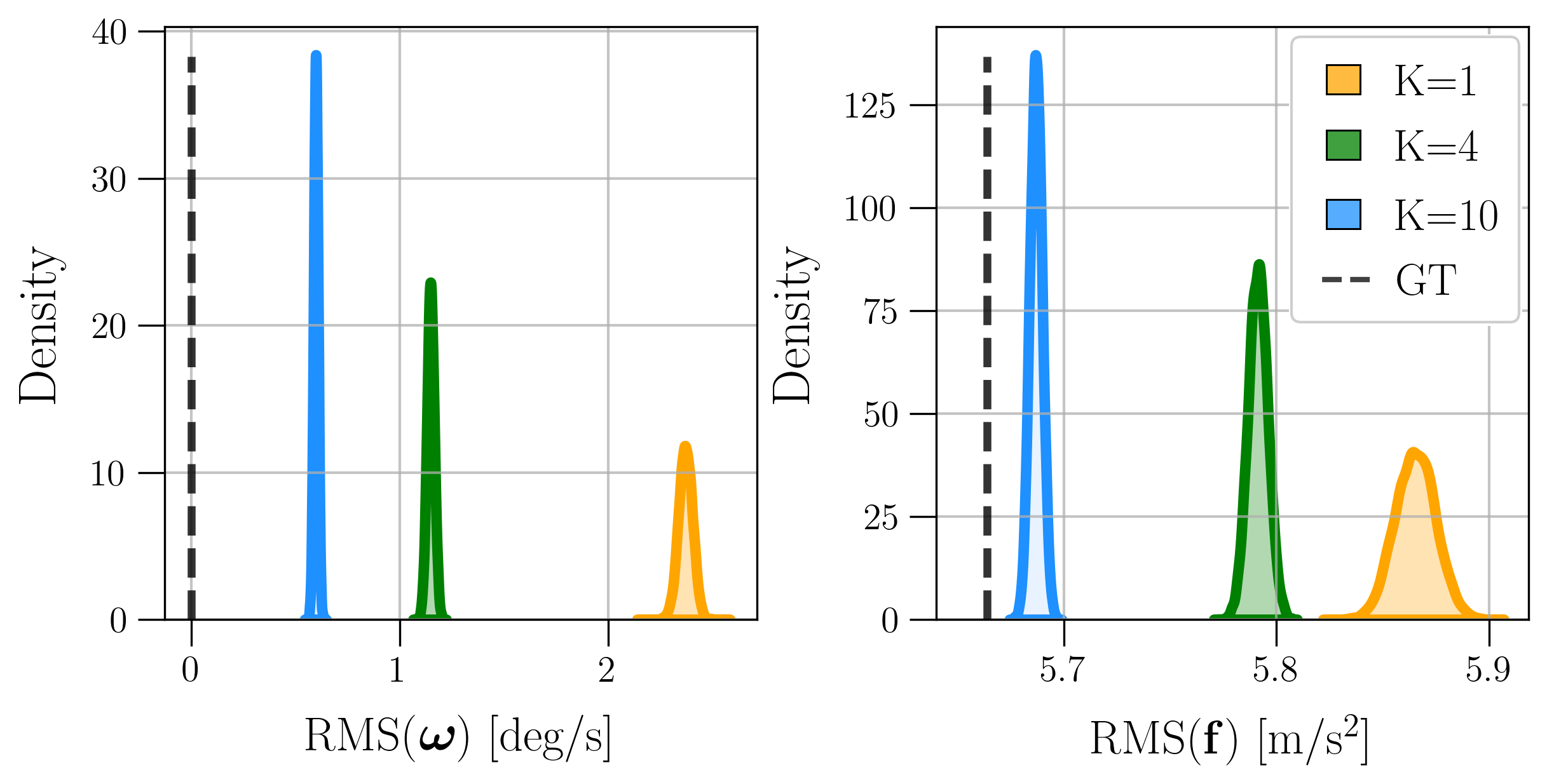}
\caption{Raw biased distributions of the inertial measurements;\\Left: gyroscopes. Right: accelerometers.}
\label{fig:PDF_Uncalibrated}
\end{figure}
Fig.~\ref{fig:Meas_Uncalibrated} illustrates the raw sensor measurements, compared over different numbers of sensors. As more sensors are added, here $K$=$\{1,4,10\}$, an improvement trend begins to emerge with respect to $K$. The thickness of each time series, denotes its noisiness, and the vertical distance from the GT dashed line expresses the bias. Furthermore, sudden outbursts, such as outliers, tend to be averaged out. Being stationary, the GT values of both sensors, are plotted in dashed lines, such that any deviation from zero is considered erroneous, except the accelerometer z-axis, whose GT satisfies $|\textbf{f}_z|=\text{g}$. 
\\
For ease of interpretation, the RMS metric is applied to the 3D signals to convey their complete magnitude. This way, the gyros RMS can simply compared with zero, as GT value. And the accelerometers GT value, is expressed by the RMS of the gravity vector by $\| \textbf{g} \| /\sqrt{3} \approx 5.664$. 
\begin{figure}[b]
\centering 
\includegraphics[width=.5\textwidth, clip, keepaspectratio]{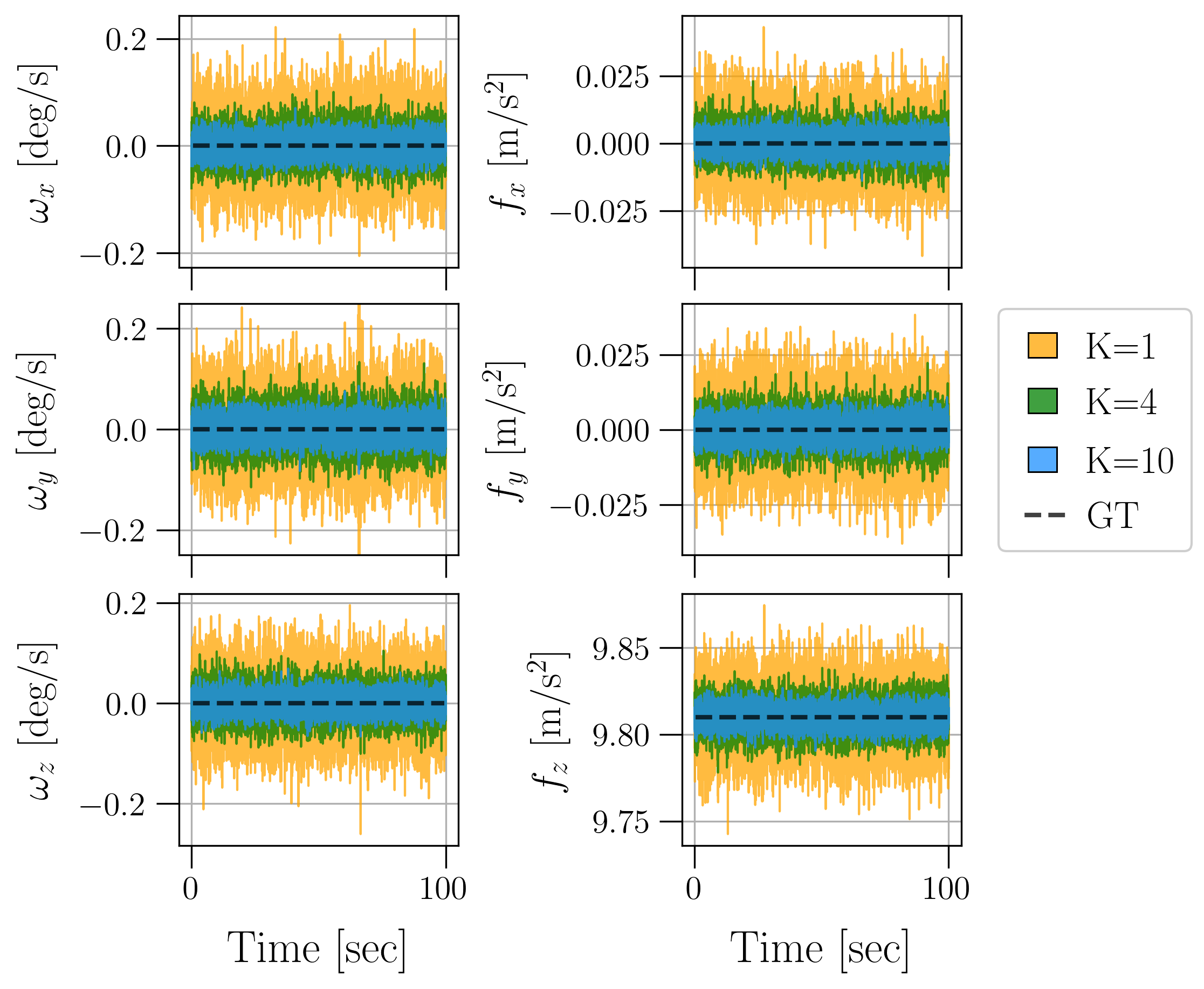}
\caption{Bias-free measurements of the inertial sensors;\\Left: gyroscopes. Right: accelerometers.}
\label{fig:Meas_Calibrated}
\end{figure}
\\
Fig.~\ref{fig:PDF_Uncalibrated} verifies the averaging effects from a distribution point of view, using the kernel density estimation (KDE) method. Unlike the probability density function (PDF), the KDE is not constrained by normalization requirements, thus can take on values greater than one.
\\
Fig.~\ref{fig:Meas_Calibrated} exploits the unbiasedness of the estimators to plot a bias-free version of the estimators, which sustains our attention on the noise mitigation. This idea is at the basis of the calibration procedure, as it aims to isolate the systematic bias from any random background noise. Once compensated, the recentered distributions are easily distinguishable by their thickness, as higher $K$'s result in higher impact. 
\begin{figure}[h]
\centering 
\includegraphics[width=.5\textwidth, clip, keepaspectratio]{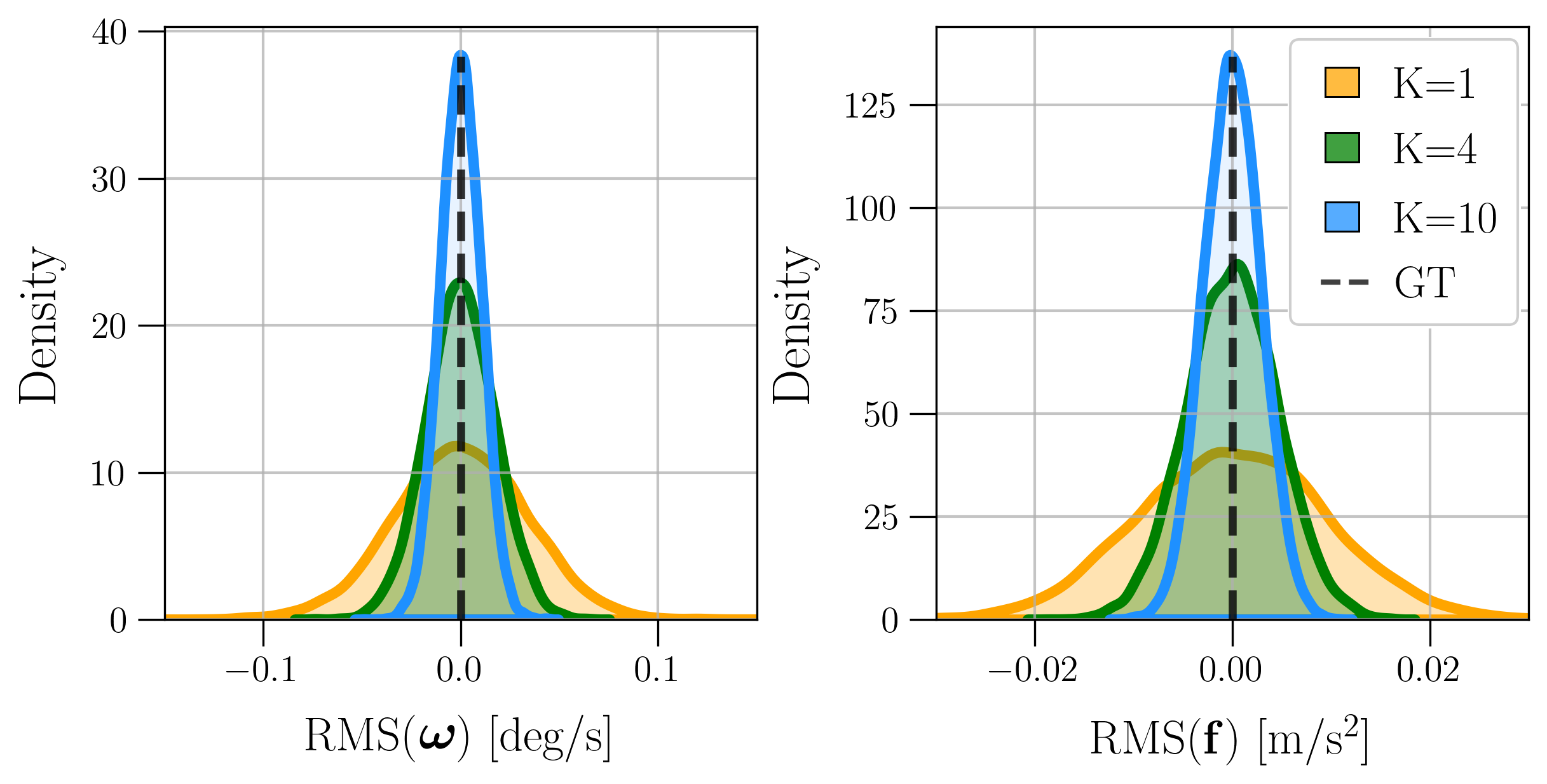}
\caption{Bias-free distributions of the inertial measurements;\\Left: gyroscopes. Right: accelerometers.}
\label{fig:PDF_Calibrated}
\end{figure}
\\
Fig.~\ref{fig:PDF_Calibrated} uses the KDE analysis once again, to visualize the reduced spread of each distribution, inversely to $\sqrt{{K}}$. The general trend becomes even clearer; the more sensors there are, the better the extracted estimates. Narrower variance results in higher precision, as noise variation is reduced. Similarly, smaller bias errors, guarantee higher accuracy, as observations fall closer to the GT. 
\\
Up to this point, the estimates have been exclusively calculated along the sensor axis. Next, the analysis is extended by averaging over the time axis, using a growing time window.
\begin{figure}[b]
\centering 
\includegraphics[width=.5\textwidth, clip, keepaspectratio]{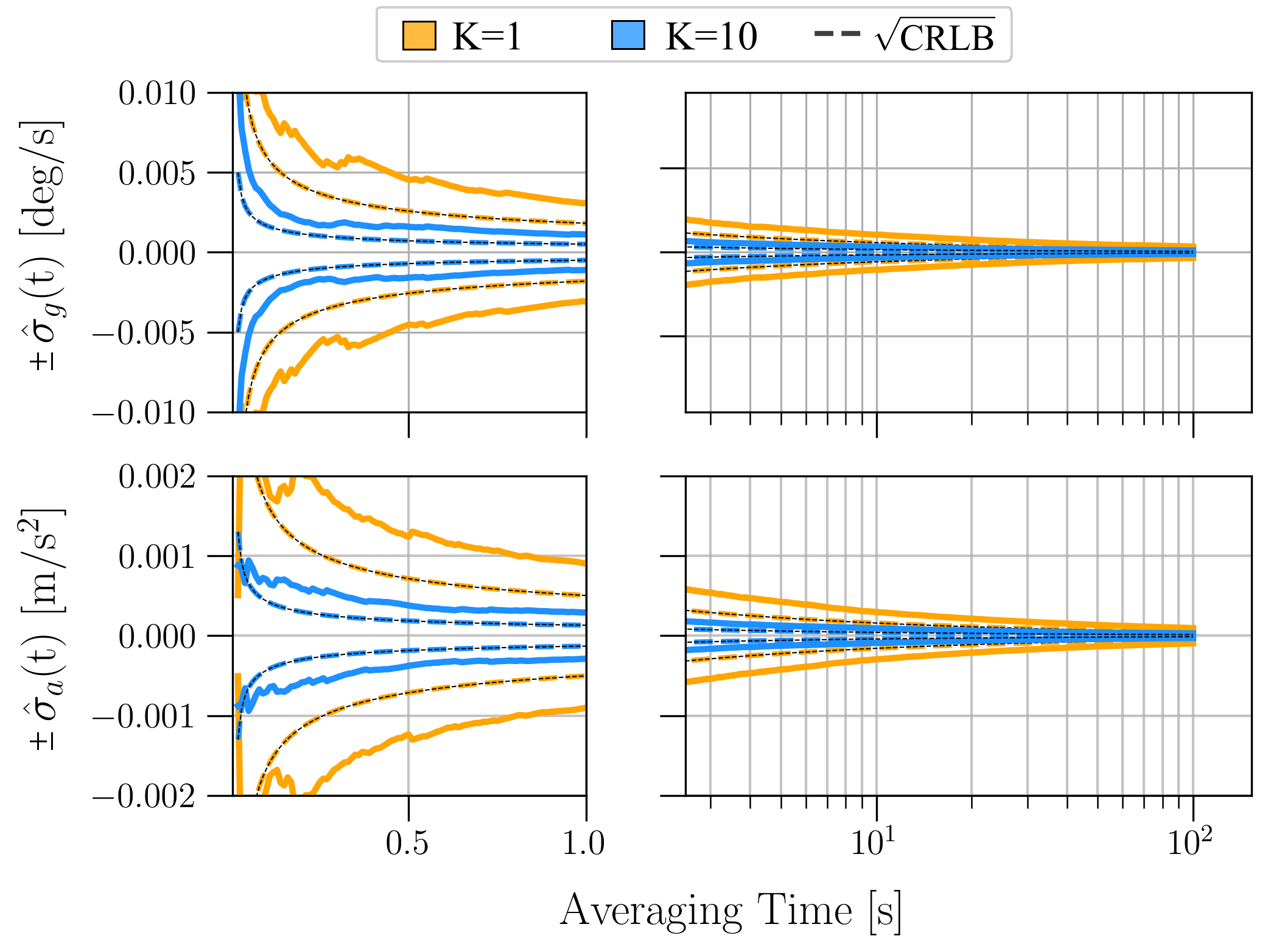}
\caption{Estimated noise density (RMS) in a semi-log plane.}
\label{fig:STD_vs_time}
\end{figure}
\\
Fig.~\ref{fig:STD_vs_time} presents the estimates of the sample standard deviation, as a function of the averaging time. To facilitate visualization, only $K=\{1,10\}$ are considered from here on. 
\\
To address the asymptotic behavior, the horizontal timeline is split into two time scales. From the left, a time resolution of one second highlights the instantaneous noise mitigation. From the right, the error is plotted on a semi-log scale to capture the exponentiality over the entire time range, implying for the marginal benefit.
Being unbiased, the estimates can be thought of as a moving standard deviation, which decreases with respect to the amount of measurements over time. These findings align well with Prop.~\ref{Prop:II}, stating that the error variance responds to both $N$ and $K$ and preserves their proportions. 
\begin{figure}[h]
\centering 
\includegraphics[width=.5\textwidth, clip, keepaspectratio]{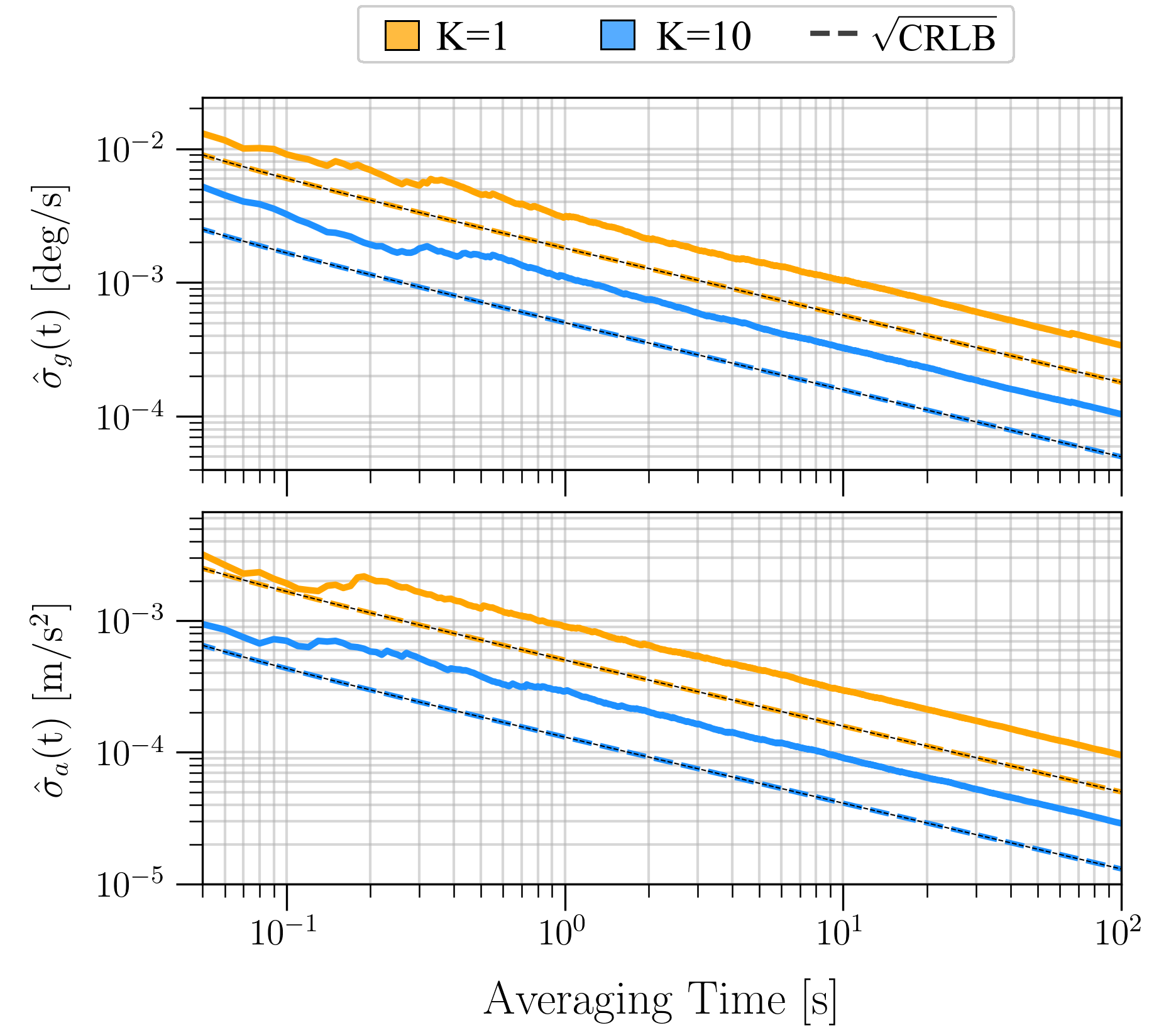}
\caption{Estimated noise density (RMS) in a log-log plane.}
\label{fig:STD_vs_time_log}
\end{figure}
\\
Fig.~\ref{fig:STD_vs_time_log} represents the estimated statistics in a logarithmic plane, such that equal distances indicate equal ratios. Using the base 10 logarithm, scaling properties can be more easily analyzed in terms of order of magnitude (OoM), i.e., a factor of 10.
\\
As shown, all slopes are equal to -0.5, thus are parallel to each other over different $K$'s. Using the square root of the CRLB \eqref{eq:CRLB}, estimates from both the individual sensor and the entire array can be compared to their minimum theoretical limits.
It is noteworthy that the unstable fluctuations observed at short timescales can be explained by the limited sample sizes, which are more sensitive to outliers. 
\\
But before reporting and analyzing the experimental results, let us recall the expected analytical behavior that relates the estimates to the independent variables, in Table~\ref{t:Params}.
\begin{table}[h]
\centering
\caption{Conceptual evaluation matrix.}
\renewcommand{\arraystretch}{2.2}
\begin{tabular}{c|c|c||c|}
\multicolumn{1}{c}{} & \multicolumn{2}{c}{\textbf{Sensor axis} $\rightarrow$} & \multicolumn{1}{c}{} \\ \cline{2-4}
\multirow{2}{*}{\rotatebox[origin=c]{90}{$\leftarrow$ \textbf{Time axis} \ }} & $\hat{\sigma}_{s, t_0}$ & $\hat{\sigma}_{m, t_0}$ & 1/$\sqrt{{K}}$ \\[1mm] \cline{2-4}
& $\hat{\sigma}_{s, t_f}$ & $\hat{\sigma}_{m, t_f}$ & 1/$\sqrt{{K}}$ \\[0mm] \cmidrule{2-4}  \morecmidrules\cmidrule{2-4} 
 & 1/$\sqrt{{N}}$ & 1/$\sqrt{{N}}$ & 1/$\sqrt{{NK}}$ \\ \cline{2-4}
\end{tabular} \label{t:Params}
\end{table}
\\
The horizontal axis distinguishes between the sample mean of a single sensor ($s$), and that of the multiple ($m$). The vertical axis denotes the time, where subscript $t_0$ refers to the initial sample at the first time step, and subscript $t_f$ denotes the accumulated sample after 100 seconds. To examine the improvement obtained in each axis, the ratios between the cells are given in the external cells, separated by $\|$. The bottom right cells denote the expected ratio between the diagonal cells.
\\
Table~\ref{t:Statistics} follows this structure, substituting the experimental values in their corresponding cells, such that their improvement ratios can be understood.
\begin{table}[h]
\captionsetup{justification=centering}
\caption{Evaluation matrix of ratios of parameters.}
\renewcommand{\arraystretch}{2.}
\begin{subtable}[c]{0.225\columnwidth}
    \begin{tabular}{|c|c||c|}
    \multicolumn{3}{c}{Gyroscopes $\sigma_g$ [deg/s]} \\ \hline
    0.0381 & 0.0121 & 0.3194 \\[1mm] \hline
    3.31\text{\sc{e}-}4 & 1.04\text{\sc{e}-}4 & 0.3131 \\[0mm] \cmidrule{1-3}  \morecmidrules\cmidrule{1-3} 
    8.71\text{\sc{e}-}3 & 8.54\text{\sc{e}-}3 & 2.72\text{\sc{e}-}3 \\ \hline
    \end{tabular}
\end{subtable}
\hspace{2.5cm}
\begin{subtable}[c]{0.225\columnwidth}
    \begin{tabular}{|c|c||c|} 
    \multicolumn{3}{c}{Accelerometers $\sigma_a$ [m/s$^2$]} \\ \hline
    0.0091 & 0.0028 & 0.313 \\[1mm] \hline
    9.60\text{\sc{e}-}5 & 3.08\text{\sc{e}-}5 & 0.321 \\[0mm] \cmidrule{1-3}  \morecmidrules\cmidrule{1-3} 
    0.0105 & 0.0108 & 3.38\text{\sc{e}-}3 \\ \hline
    \end{tabular}
\end{subtable}
\label{t:Statistics}
\end{table}
\\
Regardless of the sensor type, it can be observed that similar proportions are obtained in both tables. That is, the ratios calculated in the rightmost column and the bottom row converge into relatively close constants. Since every increase or decrease of ten decibels [dB] represents a tenfold change, consider the following interpretation:
\begin{align} \label{eq:log_10} 
10 \log_{10} (\sim 0.31) &\approx -5 \hspace{3mm} \text{[dB]} \ \Leftrightarrow \ { {K}_{10}}/{ {K}_{1} } = \frac{1}{\sqrt{{K}}} \ , \\
10 \log_{10} (\sim 9\text{\sc{e}-}3) &\approx -20 \ \text{[dB]} \ \Leftrightarrow \ { t_f/t_0 } = \frac{1}{\sqrt{{N}}} \ .
\end{align}
%
In layman's terms, given $K$ sensors, the standard deviation of the error estimate is reduced by the square root, i.e., $1/\sqrt{10}$. Applied between successive time steps, sensor averaging does not depend on time (see parallel slopes), thus can be also used in dynamic conditions, without inducing latency. 
In contrast, the total number of measurements does depend on time, growing by a factor of the sampling rate. Here, given 100-second measurements sampled at 100 Hz, the reduction ratio decreases by the total number of instances, i.e., 1/$\sqrt{10^{-4}}$. 
\\
To conclude, the asymptotic behavior of the sample means was investigated with respect to the number of measurements, obtained from either of the sensors or the amount of time. The potential benefit lies in the following trade-off: a tenfold increase in the averaging time, i.e., a decade, reduces the error by -5 [dB/dec]. Similarly, multiplying the number of sensors by ten, reduces the error by -5 [dB/arr$_{10}$]. 
%
%
\subsection{State estimation}
So far, the underlying measurement model was a priori assumed, with a fixed, yet unknown parameters.
%
In this section however, the calculated estimates are stochastic, as they are intended to reveal the hidden states only with respect to a specific point in time. Here, the observed dynamics constitute a navigation problem, whose Gaussian nature enables extrapolating both its means and uncertainties, forward in time. 
\\
Driven by the uncompensated biases, the error states indicate the discrepancy between the estimates and the true states. Their noisiness, is given by the associated error state covariances, propagated forward in proportion to their initial estimates. 
\\
In the absence of any fusion mechanism, each state of the INS solution is free to drift in time, expressing the error intensities with respect to the true states, which remain constant. As a result, the CRLB can no longer be applied, as all estimated states are biased for any $t>0$.
\\
To calculate the discrete measurements with the theory developed in Section~\ref{sec:theory}, the relevant terms are discretized to allow a piecewise approximation. Under LTI dynamics, the discrete-time state-transition matrix is given by
\begin{align}
\mathbf{\Phi}_k = \mathbf{\Phi} = e^{ \mathbf{F} \Delta t } \ , 
\end{align}
exhibiting shift invariance, such that subscript $k$ is omitted, letting the discrete error state vector by
\begin{align} \label{eq:errorsDiscrete} 
\delta {\mathbf{x}}_{k+1} &= \mathbf{\Phi} \, \delta {\mathbf{x}}_{k} \ . 
\end{align}
Fig.~\ref{f:errorStates} presents a squared 3$\times$3 image, where all 3D kinematic errors are arranged row by row: position error, velocity error, and misalignment error. Each sub-image denotes a directional component describing the resulting magnitude over time as a function of the number of sensors. 
\begin{figure}[h]
\centering 
\includegraphics[width=.5\textwidth, clip, keepaspectratio]{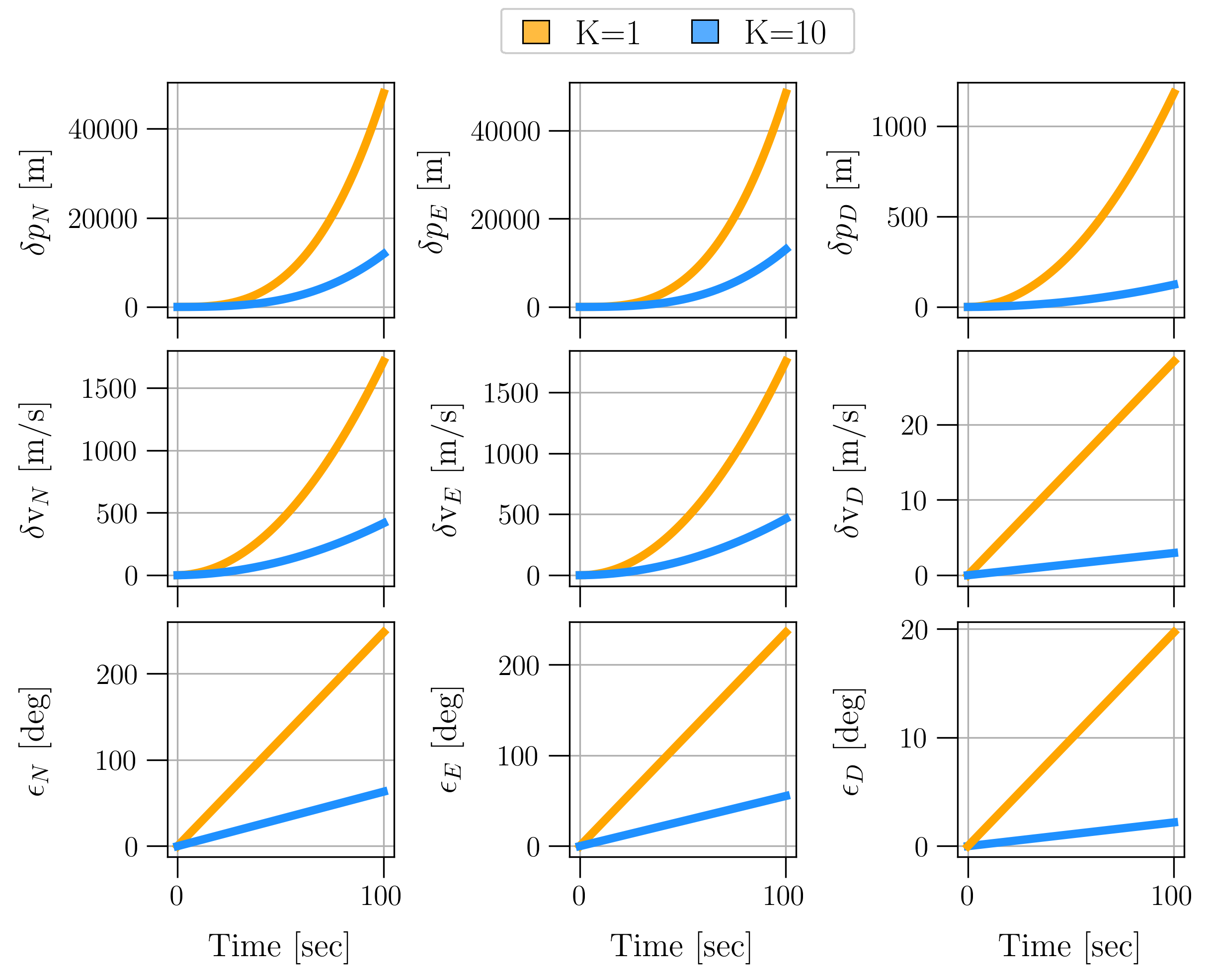}
\caption{Propagation of the error states, comparing single sensor and sensor array, in absolute values.}
\label{f:errorStates}
\end{figure}
\\
Overwhelmingly, all of the error states estimated from an array of ten sensors outperform those estimated from a single sensor. As expected, the position errors present the strongest error growth, apparently due to the consecutive integrations in time of the biases.
\\
Going down the rows, this fact is supported visually as well, as the polynomial growth becomes linear. 
While the two left columns exhibit same scale values, the rightmost column shows values that are smaller by 1 OoM. This is explained by the skew-symmetric matrix \eqref{eq:skewSym}, whose gravity projections ($\approx$10) fall only on the north-east plane.
\\
Next, the discrete-time error state covariance is given by
\begin{align}
\mathbf{P}_{k+1} &= \mathbf{\Phi} \, {\mathbf{P}_{k}} \, \mathbf{\Phi}^{\TT} + \mathbf{Q}_k \ . 
\end{align}
Being completely static, the true states are known at all times, such that the initial uncertainties are negligible, i.e., $\mathbf{P}_{0}=0$. 
\begin{figure}[t]
\centering 
\includegraphics[width=.5\textwidth, clip, keepaspectratio]{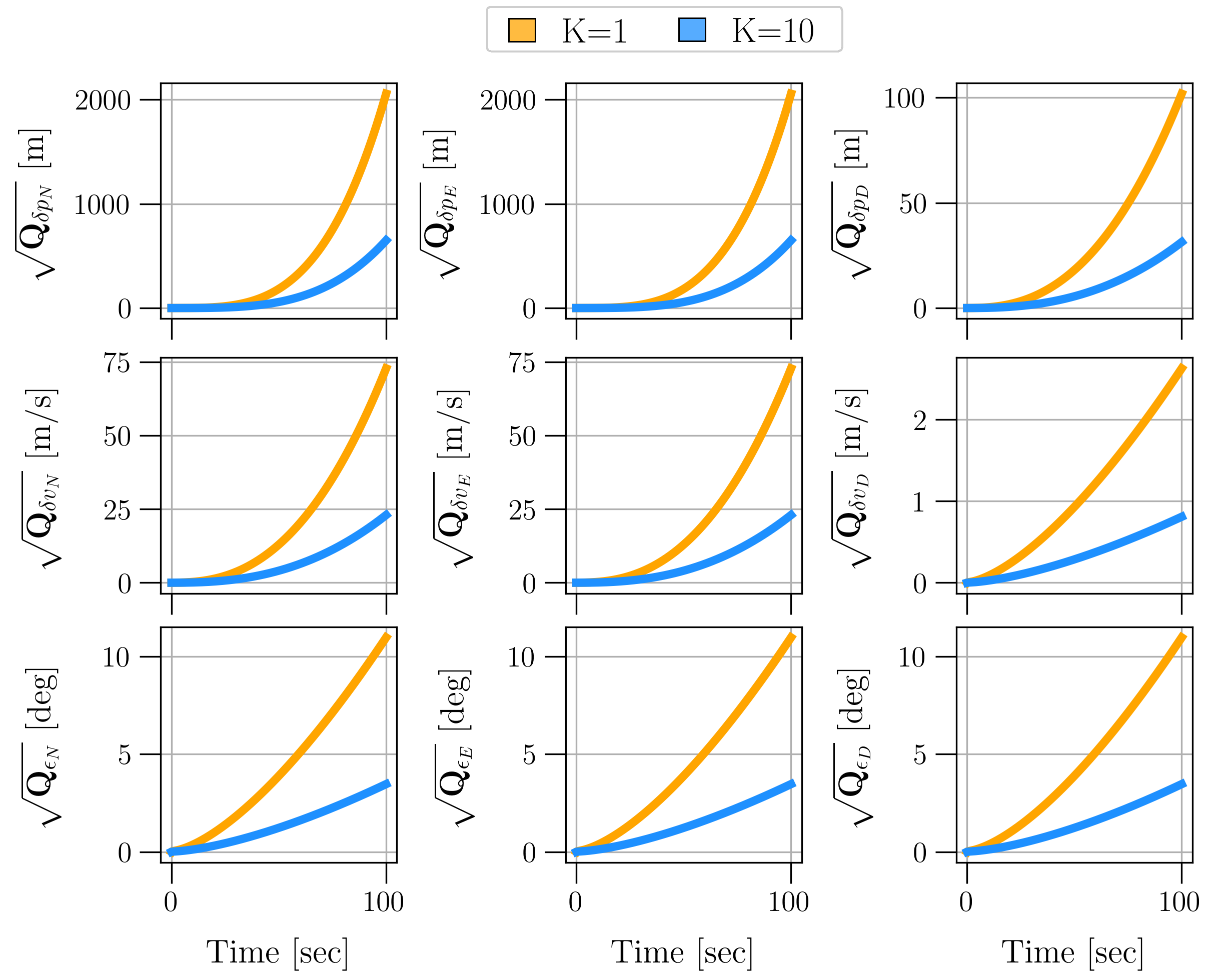}
\caption{Propagation of the state uncertainties, comparing single sensor and sensor array.}
\label{f:fig_cov}
\end{figure}
Fig.~\ref{f:fig_cov} describes the state uncertainties, calculated iteratively through the square root diagonal of discrete \eqref{eq:CovarFull}. Following similar patterns as before, the state uncertainties of the single sensor grow faster than those estimated from the entire array. 
\\
An important point to dwell on is that the state estimates and their uncertainties do not lie on the same scale, as state uncertainties are significantly smaller. 
%
The explanation for this lies in the instrumental errors of the given sensors. There, the standard deviations of the noise sources are significantly smaller than those of their counterpart biases, i.e., $\| \boldsymbol{\sigma} \| \ll \| \textbf{\textit{b}} \|$. This even worsens once they are squared to obtain their variance. As a result, the uncertainties observed converge into smaller numbers by roughly 1 OoM, i.e., $\sqrt{\hat{\mathbf{Q}}_{jj}} \approx 0.1 | \delta \mathbf{x}_j |$. 
%
\\
Table~\ref{t:Stats_comp} summarizes the numbers from both figures above, following the same 3$\times$3 structure. To answer the bottom line question of effectiveness, it measures the K$_{10}/$K$_{1}$ ratio between each state, subscript j, after 100 seconds. The left table reports the reduction ratio of each state estimate, when estimates from ten sensors are divided by estimates from a single sensor. Similarly, the right table reports how much their corresponding state uncertainties are reduced.
\begin{table}[h]
\captionsetup{justification=centering}
\caption{Evaluation matrix of ratios of estimates.}
\renewcommand{\arraystretch}{2.}
\hspace{4mm}
\begin{subtable}[c]{0.225\columnwidth}
    \begin{tabular}{|c|c|c|}
    \multicolumn{3}{c}{$\EV[ \delta \hat{\mathbf{x}}_{m,j} ] \ / \  \EV[ \delta \hat{\mathbf{x}}_{s,j} ]$} \\ \hline
    0.2473 & 0.2690 & 0.1045 \\[1mm] \hline
    0.2432 & 0.2649 & 0.1014 \\[1mm] \hline 
    0.2542 & 0.2342 & 0.1109 \\ \hline
    \end{tabular}
\end{subtable}
\hspace{2.2cm}
\begin{subtable}[c]{0.225\columnwidth}
    \begin{tabular}{|c|c|c|}
    \multicolumn{3}{c}{$\sqrt{ \hat{\mathbf{Q}}_{m, jj} \ / \  \hat{\mathbf{Q}}_{s, jj} }$} \\ \hline
    0.3157 & 0.3151 & 0.3077 \\[1mm] \hline
    0.3149 & 0.3158 & 0.3081 \\[1mm] \hline 
    0.3156 & 0.3159 & 0.3162 \\ \hline
    \end{tabular}
\end{subtable} \label{t:Stats_comp}
\end{table}
\\
The right table exhibits a high degree of agreement with the anticipated analysis in Prop.~\ref{Prop:V}, as most values converge in the vicinity of $\sqrt{ \hat{\mathbf{Q}}_{m, jj}} = \sqrt{\hat{\mathbf{Q}}_{s, jj}/10 }$. In contrast, empirical results in the left table do not fully overlap with the analysis stated in Prop.~\ref{Prop:IV}, i.e., $\EV[ \delta \hat{\mathbf{x}}_{m,j} ] = \frac{1}{10} \EV[ \delta \hat{\mathbf{x}}_{s,j} ]$. 
\\
This partial agreement occurs due to the skew-symmetric matrix \eqref{eq:skewSym}, which multiplies only the north-east components by gravity, thus obscuring the natural proportions between the sensors. However, it is important to acknowledge that such discrepancies do not necessarily indicate a modelling error, as our Propositions constitute generalized statements, independent of any specific dynamics.
\begin{figure}[t]
\centering 
\includegraphics[width=.5\textwidth, clip, keepaspectratio]{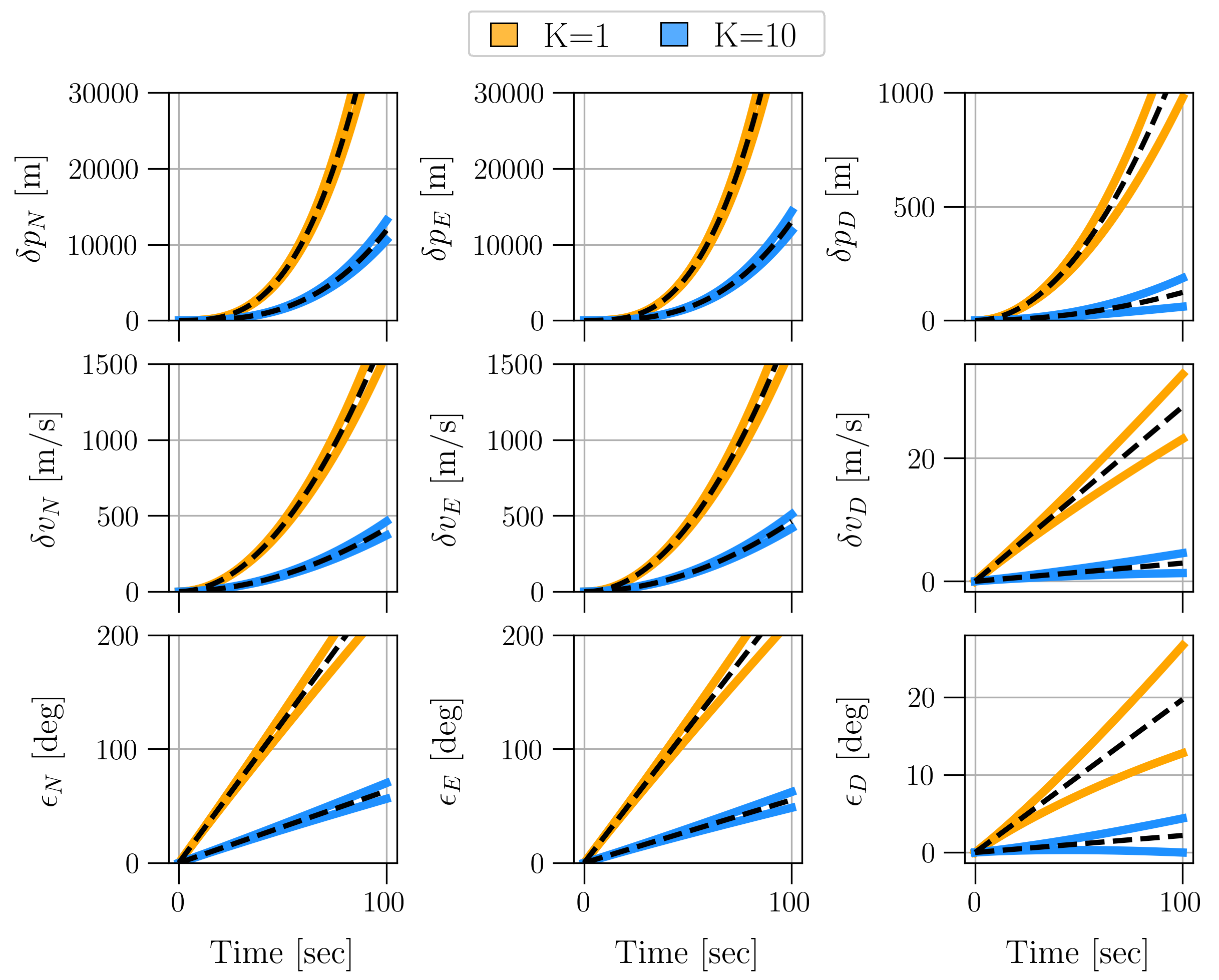}
\caption{Time-varying distribution; state estimates in black dashed lines, accompanied by their uncertainties, in colors.}
\label{f:fig_comb}
\end{figure}
\\ 
Fig.~\ref{f:fig_comb} combines both figures above, emphasizing the error model behavior as a linear Gaussian process, fully characterized by a mean vector and a state covariance matrix. To provide some intuitive interpretation, the evolution of the covariance matrix can be represented visually as a 3D ellipsoid, using its calculated eigenvalues and eigenvectors to determine the principal axes and their orientation, respectively. 
\\
\begin{figure}[h]
\centering 
\includegraphics[width=.52\textwidth, clip, keepaspectratio]{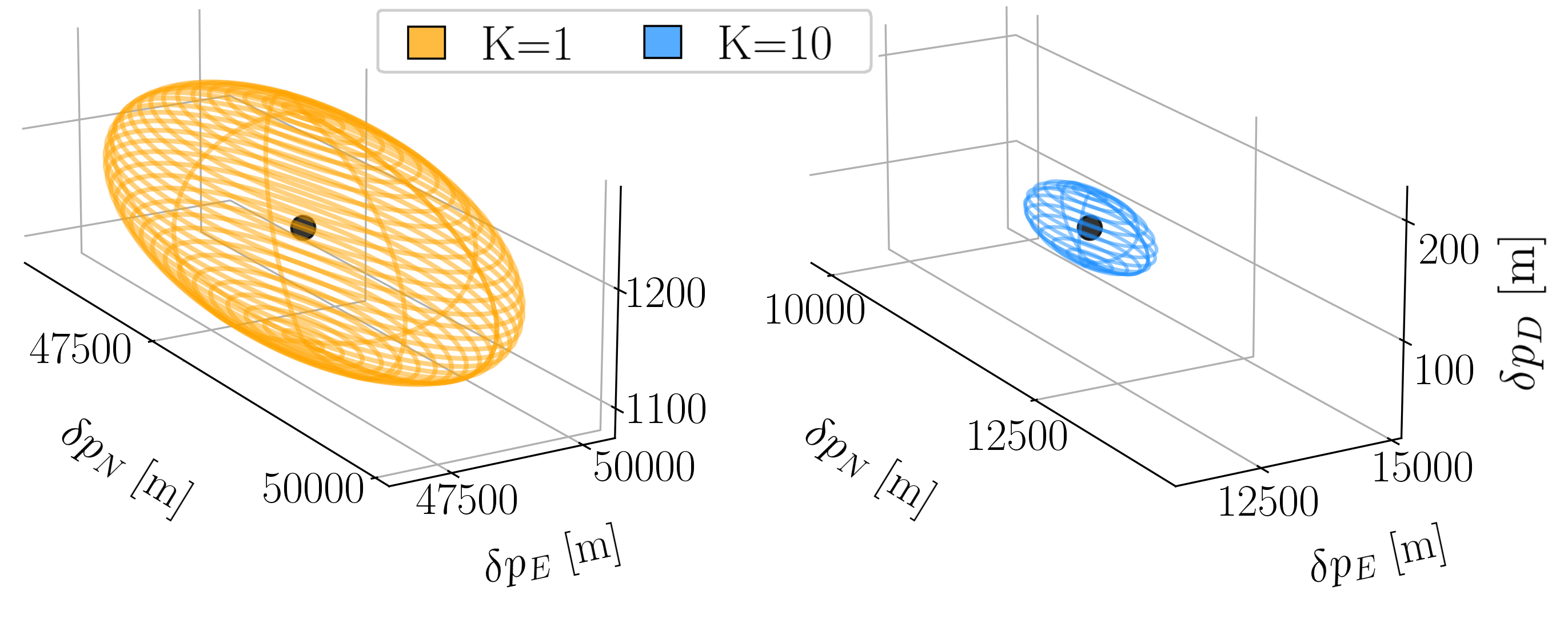}
\caption{Error ellipsoid of position error after 100 seconds.}
\label{f:fig_ellipsoid}
\end{figure}
\\
Fig.~\ref{f:fig_ellipsoid} offers a fair spatial comparison between the propagated errors within 100 seconds, as both span over equal distances in the 3D space. The centroid of each ellipsoid, $\delta\mathbf{p}^{n}$, denotes the 3D error vectors, and their outer surfaces, $\sqrt{\mathbf{Q}_{\delta\mathbf{p}^n}}$, denote the confidence region within one standard deviation.
\\
By examining the centroid values and their envelope volume, it is clearly seen that by incorporating more sensors, both the errors and their uncertainties are dramatically reduced.

\section{Discussion} \label{sec:disc}
Having thoroughly processed the results, we employ Fig.~\ref{f:accuracy} to illustrate and emphasize three of the main insights stemming from our study (in arbitrary units):
\begin{enumerate}
    \item Consistency: as additional sensors are integrated, estimates steadily converge towards the true parameter being estimated, enhancing the reliability of the results.
    \item Accuracy and precision: by averaging over multiple sensors, the estimated sample variance is reduced, suggesting an improved precision. Similarly, biases from multiple sensors cancel out each other, bringing the sample mean closer to the GT, i.e. enabling better accuracy.
    \item Reduced sensitivity: increasing the overall number of measurements reduce the exposure of the estimates to instrumental errors and outliers, thus leading to a smoother probability density function.
\end{enumerate}

\begin{figure}[h]
\centering 
\includegraphics[width=.45\textwidth, clip, keepaspectratio]{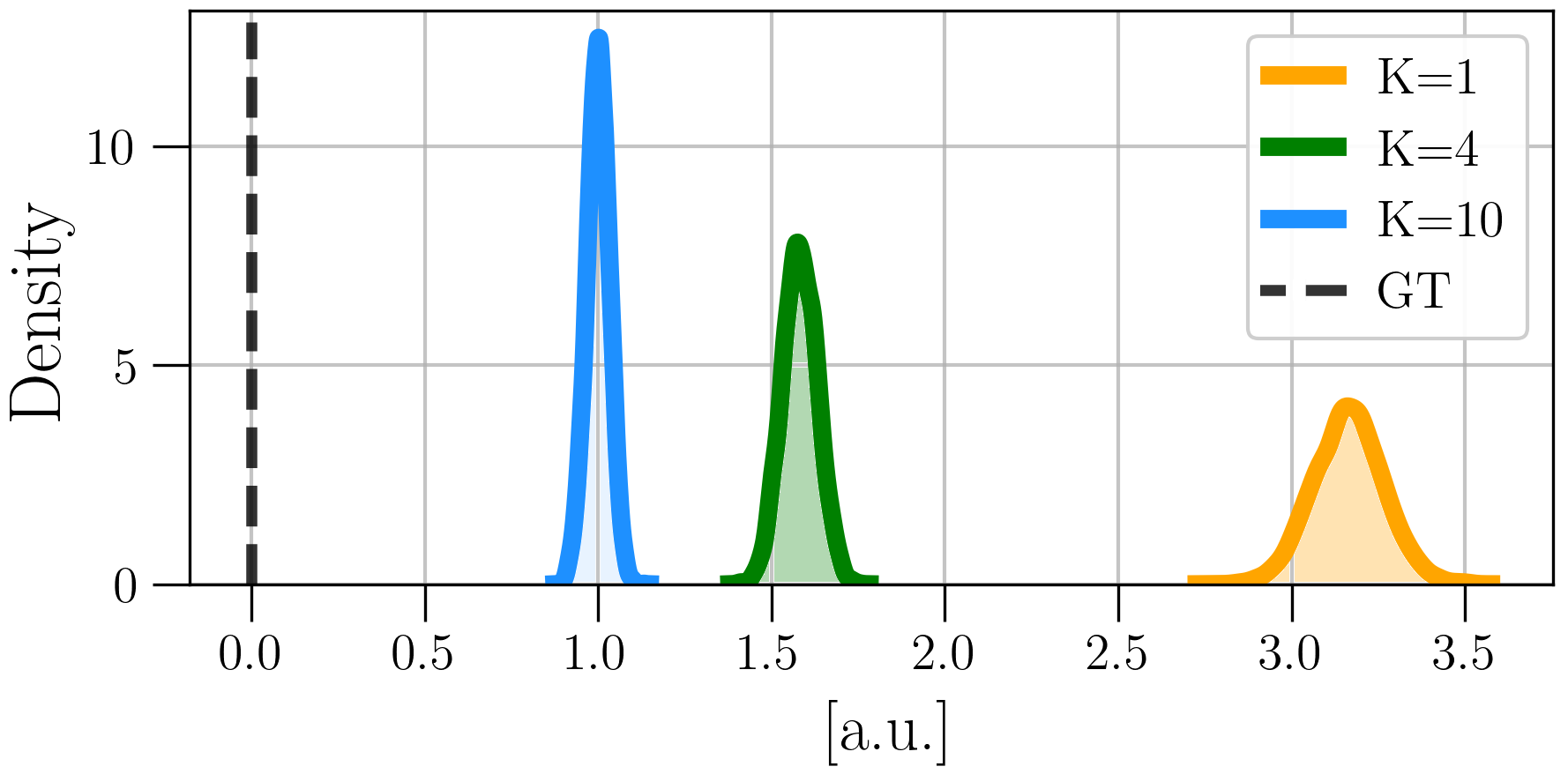}
\caption{Effects of sensor count on precision and accuracy.}
\label{f:accuracy}
\end{figure}
%
\subsection{Limitations of this study}
While the study provides valuable insights into the use of multiple sensors, several limitations should be acknowledged: 
\begin{enumerate}
    \item No variability: the sensors used in this study originated from the same production line. Their nearly identical error range facilitated not only the model analysis, but also streamlined the practical implementation.
    \item No dynamics: to sustain stationarity of the random errors, the study focused exclusively on stationary conditions across the entire recordings. This allowed remaining consistent with the closed-form analysis, without violating the zero-mean Gaussian assumption. 
    \item No SNR: A commonly used metric in similar studies, especially when referring to the CRLB metric, is the signal-to-noise-ratio (SNR). Often used as the independent variable, it is able to embody the signal strength in relation to the noise floor. However under the stationary conditions discussed, the SNR was consciously omitted due to its potential for misinterpretations with accuracy.
    \\
    For example, consider a stationary gyroscope whose expectation should embody the pure bias only. When using the SNR, higher levels of bias exhibit higher SNR, instead of reflecting the degraded accuracy. 
\end{enumerate}

\subsection{Future research}
In light of the promising outcomes of the study, while recognizing its limitations, we propose several research directions for further explorations:
\begin{enumerate}
    \item Error diversity: exploring fusibility between sensors of different error scales, both analytically, e.g., Fisher Information and CRLB, and empirically.
    \item Non-stationary conditions: extending the analytical model into dynamic scenarios where the angular velocities and the accelerations inputs are not assumed to be zero.
    \item Optimal INS integration: investigating which filtering method best fits with the multiple IMU settings, e.g., federated KF, ensemble KF, augmented KF, or the virtualization process.
\end{enumerate}


\section{Conclusion} \label{sec:conc}
This study aimed at investigating the relationship of the sample error to the sampling duration, and the number of sensors used. It was carefully designed to remain user-friendly, and simplify complex estimation topics, thus being suitable for both entry-level learners and advanced practitioners alike. 
At first, an analytical framework was developed to model the estimation error, starting from the signal level, all the way up to the entire INS solution. 
%
Upon cleansing the experimental outcomes, empirical data was closely examined and then compared with the anticipated results.
\\
In conclusion, our findings demonstrate a high level of agreement with the hypothesized model, thereby substantiating its validity and providing meaningful practical implications. Adding more sensors can indeed act as a force multiplier, as it expedites the estimation process of the true system states.
\\
Ultimately, these findings reach deep into the design of high-precision systems, which often rely on the standalone INS. 
\\
For example, during in-field deployment, non-sterile conditions are likely to contaminate the samples, thereby increasing the required averaging times. However, by adding more sensors, these operational constraints can be better met, thus guaranteeing a faster initialization at the cost of minor changes.
\appendix 
\subsection{Wide sense stationarity} \label{appendix:a}
A discrete stochastic process is said to be wide-sense stationary (WSS) if its statistical properties do not change over time. Sampled under stationary conditions, let the following realization of sensor noise 
\begin{align}
\EV[ \textit{w}[n] ] = \EV[ \textbf{\textit{w}} ] \quad & \forall \quad n \in \ \mathbb{Z} \ , \\
\operatorname{R}_{ww}(n, k) \triangleq \sum_n \textit{w}[n] \textit{w}[n-k] \quad & \forall \quad n, k \ \in \ \mathbb{Z} \ ,
\end{align}
where $\EV$ and $\operatorname{R}_{xx}$ denote the mean and the autocorrelation function (ACF), respectively, $n$ denotes the discrete time index, and $k$ is the time lag. 

\subsection{Error state covariance} \label{appendix:b}
Let the continuous process noise covariance be
\begin{align*}
\mathbf{Q}(\tau) = \EV[ \mathbf{G} \textbf{\textit{w}} \textbf{\textit{w}}^{\TT} \mathbf{G} ] = \int_{t_0}^t \mathbf{\Phi}(\tau) \mathbf{G} \, \mathbf{S}_{xx} \, \mathbf{G}^{\TT} \mathbf{\Phi}(\tau)^{\TT} d\tau \ .
\end{align*} 
Under the reasonable assumption that the sensor noise is isotropic, $\mathbf{I}_{3} \sigma^2$ denotes spherical Gaussian distribution with equal variance in all three directions, such that 
\begin{align} \label{eq:CovarFull}
& \mathbf{Q}(\tau) = \\[2mm]
& \left[ \arraycolsep=2.5pt \def\arraystretch{2.2}
\begin{array}{ccccc} 
\mathbf{Q}_{ \mathbf{p}\mathbf{p} } & \mathbf{Q}_{ \mathbf{p}\mathbf{v} } & \mathbf{Q}_{ \mathbf{p}\mathbf{\epsilon} } & \frac{\sigma_{ab}^2 \tau^3  }{6} \mathbf{I}_{3} & \frac{\sigma_{gb}^2 \tau^4  }{24} \mathbf{F}_{23} \\
\mathbf{Q}_{ \mathbf{p}\mathbf{v} }^{\TT} & \mathbf{Q}_{ \mathbf{v}\mathbf{v} } & \mathbf{Q}_{ \mathbf{v}\mathbf{\epsilon} } & \frac{\sigma_{ab}^2 \tau^2 }{2} \mathbf{I}_{3} & \frac{\sigma_{gb}^2 \tau^3  }{6} \mathbf{F}_{23} \\
\mathbf{Q}_{ \mathbf{p}\mathbf{\epsilon} }^{\TT} & \mathbf{Q}_{ \mathbf{v}\mathbf{\epsilon} }^{\TT} & \mathbf{Q}_{ \mathbf{\epsilon}\mathbf{\epsilon} } & \mathbf{0}_{3} & \frac{ \sigma^2_{gb} \tau^2}{2} \mathbf{I}_{3}  \\
\frac{\sigma_{ab}^2 \tau^3  }{6} \mathbf{I}_{3} & \frac{\sigma_{ab}^2 \tau^2 }{2} \mathbf{I}_{3} & \mathbf{0}_{3} & \sigma^2_{ab} \tau \mathbf{I}_{3}  & \mathbf{0}_{3} \\
\frac{\sigma_{gb}^2 \tau^4  }{24} \mathbf{F}_{23}^{\TT} & \frac{\sigma_{gb}^2 \tau^3  }{6} \mathbf{F}_{23}^{\TT} & \frac{ \sigma^2_{gb}\tau^2 }{2} \mathbf{I}_{3}  & \mathbf{0}_{3} & \sigma^2_{gb} \tau\mathbf{I}_{3} \\
\end{array} 
\right] \ . \notag
\end{align}
For brevity, $\mathbf{F}_{23}^{\wedge{}} = \mathbf{F}_{23} \mathbf{F}_{23}^{\TT}$ denotes the self-adjoint of the gravitation matrix. The diagonal elements, i.e., $\mathbf{Q}_{ii}$, denote the variability in time of each state, 
\begin{align*}
\mathbf{Q}_{ \mathbf{p}\mathbf{p} } &= \frac{\sigma_{gb}^2 \mathbf{F}_{23}^{\wedge{}}  }{252} \tau^7 + \left( \frac{ \mathbf{F}_{23}^{\wedge{}} \sigma^2_{g} + \mathbf{I}_{3}\sigma^2_{ab} }{20} \right) \tau^5 + \frac{ \sigma^2_{a}\mathbf{I}_{3} }{3} \tau^3 \ , \\ 
\mathbf{Q}_{ \mathbf{v}\mathbf{v} } &= \frac{\sigma_{gb}^2 \mathbf{F}_{23}^{\wedge{}}  }{20}\tau^5 +  \left( \frac{ \mathbf{F}_{23}^{\wedge{}} \sigma^2_{g} + \mathbf{I}_{3} \sigma^2_{ab} }{3}\right) \tau^3 + \frac{ \sigma^2_{a}\mathbf{I}_{3} }{2} \tau  \ , \\ 
\mathbf{Q}_{ \mathbf{\epsilon}\mathbf{\epsilon} } &= \frac{ \sigma^2_{gb}\mathbf{I}_{3}  }{3} \tau^3 + \sigma^2_{g} \mathbf{I}_{3} \tau \ . \\
\end{align*}
In contrast, the off-diagonal elements represent the degree to which different pairs of errors are intercorrelated:
\begin{align*}
\mathbf{Q}_{ \mathbf{p}\mathbf{v} } &= \frac{\sigma_{gb}^2 \mathbf{F}_{23}^{\wedge{}}  }{72}\tau^6 +  \left( \frac{ \mathbf{F}_{23}^{\wedge{}} \sigma^2_{g} + \mathbf{I}_{3} \sigma^2_{ab} }{8}\right) \tau^4 + \frac{ \sigma^2_{a}\mathbf{I}_{3} }{2} \tau^2 \ , \\ 
\mathbf{Q}_{ \mathbf{p}\mathbf{\epsilon} } &= \frac{\sigma_{gb}^2 \mathbf{F}_{23} }{30} \tau^5 + \frac{\sigma_g^2 \mathbf{F}_{23}}{6} \tau^3 \ ,\\ 
\mathbf{Q}_{ \mathbf{v}\mathbf{\epsilon} } &= \frac{\sigma_{gb}^2 \mathbf{F}_{23} }{8} \tau^4 + \frac{\sigma_g^2 \mathbf{F}_{23}}{2} \tau^2 \ . \end{align*}

\bibliographystyle{IEEEtran}
\bibliography{Ref}
\vspace{-1cm}

\end{document}